\documentclass[11pt, oneside]{Thesis} 

\graphicspath{{Pictures/}} 

\usepackage[square, numbers, comma, sort&compress]{natbib} 
\hypersetup{urlcolor=blue, colorlinks=true} 
\title{\ttitle} 

\begin{document}

\frontmatter 

\setstretch{.5} 

\fancyhead{} 
\rhead{\thepage} 
\lhead{} 

\pagestyle{fancy} 

\newcommand{\HRule}{\rule{\linewidth}{0.5mm}} 

\hypersetup{pdftitle={\ttitle}}
\hypersetup{pdfsubject=\subjectname}
\hypersetup{pdfauthor=\authornames}
\hypersetup{pdfkeywords=\keywordnames}

%
%
%
%
%
%
%
%
\begin{titlepage}

\center
 
\raisebox{-1.cm}[0pt][0pt]{\includegraphics[width=1.5cm]{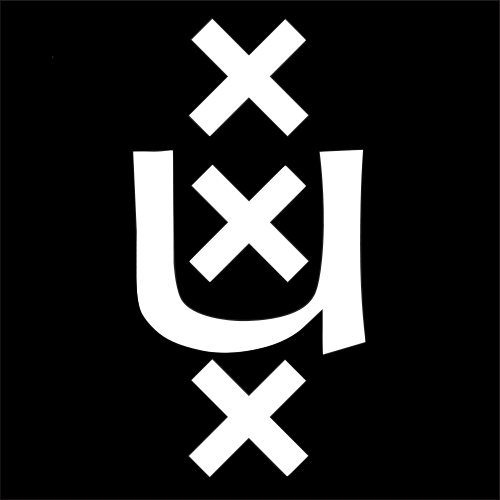}}
\raisebox{-.45cm}[0pt][0pt]{\textsc{\Huge University of Amsterdam}}\\[4.5 cm]

\Large{\textbf{MSc Computational Science}}\\[0.7cm]
\textsc{\Large \textbf{Master Thesis}}\\[1.cm]


\HRule \\[0.9cm]

{ \Large \bfseries Secure Identification in the Isolated Qubits Model}\\[0.9cm] 

\HRule \\[1.6cm]
 

{\large \emph by}\\[0.9cm]
{\large \bfseries Filippos-Arthouros Vogiatzian-Ternaxizian\\ 
10661565}\\[0.9cm] 
{\large  \emph{October 2015}}\\ 
[1.8cm] 


\begin{minipage}{0.4\textwidth}
\begin{flushleft} \large
{\large \emph{Supervisor:}}\\

\large{Dr. Christian Schaffner} 
\end{flushleft}
\end{minipage}
\hfill
\begin{minipage}{0.4\textwidth}
\begin{flushright} \large
\emph{Examiners:} \\
\large{Dr. Inge Bethke} \\ 
\large{Dr. Serge Fehr} 
\end{flushright}
\end{minipage}\\

\begin{minipage}{0.4 \textwidth}
 \begin{flushleft}
  
\raisebox{-158pt}[0pt][0pt]{\includegraphics[height=1.5cm, width=4cm]{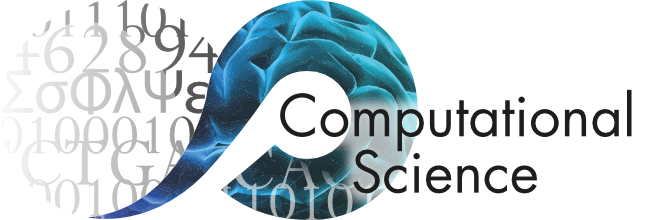}}\\ 

\raisebox{-158pt}[0pt][0pt]{\large{Informatics Institute/UvA}} 

 \end{flushleft}

\end{minipage}
\hfill
\begin{minipage}{0.4\textwidth}
 \begin{flushright}
  
\raisebox{-158pt}[0pt][0pt]{\includegraphics[height=1.5cm, width=4cm]{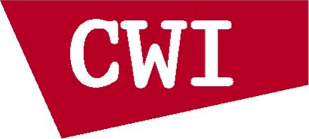}}\\ 

\raisebox{-158pt}[0pt][0pt]{\large{Centrum Wiskunde \& Informatica}} 

 \end{flushright}
\end{minipage}


\vfill 

\end{titlepage}

\addtotoc{Abstract} 

\abstract{\addtocontents{toc}{\vspace{1em}} 

Oblivious transfer is a powerful cryptographic primitive that is complete for secure multi-party computation.  
In oblivious transfer protocols a user sends one or more messages to a receiver, while the sender remains oblivious as to which messages have been received.
Protocols for oblivious transfer cannot exist in a classical or fully-quantum world, but can be implemented by restricting the users' power. 

The isolated qubits model is a cryptographic model in which users are restricted to single-qubit operations and are not allowed to use entangling operations.
Furthermore, all parties are allowed to store qubits for a long time before measuring them. 

In this model, a secure single-bit one-out-of-two randomised oblivious transfer protocol  was recently presented by Liu.
Motivated by this result, we construct a protocol for secure string one-out-of-two randomised oblivious transfer by simplifying and generalising the existing proof.

We then study for the first time interactive protocols for more complex two-party functionalities in this model based on the security of our construction. 
In order to guarantee the composability of our construction, users are restricted to measurement at the end of each sub-protocol. It is then possible to construct
secure one-out-of-two and one-out-of-k oblivious transfer protocols in the isolated qubits model.

Moreover, we study secure password-based identification, where a user identifies himself to another user by 
evaluating the equality function on their inputs, or passwords. 
We use the oblivious transfer constructions mentioned above as sub-protocols to construct a secure identification protocol.

Finally, we prove that constructing a secure identification protocol non-interactively is impossible, even using oblivious transfer.

}

\clearpage 

%
\setstretch{1.3} 

\acknowledgements{\addtocontents{toc}{\vspace{1em}} 

First of all, I would like to thank my supervisor, Christian Schaffner for introducing me to world of quantum cryptography and for giving
me the opportunity to work with him, for his valuable contribution throughout the project, the long hours he spent on trying to solve the riddles
of isolated qubits. 

Furthermore, I want to thank Yi-Kai Liu for helpful discussions and suggestions as well as reading through our first try to tackle his model.

I would also like to thank the examination committee for taking the time and effort of reading this thesis.

Last but not least, I want to thank my family and friends for their motivation and support during the last year. 
}
\clearpage 


\pagestyle{fancy} 

\lhead{\emph{Contents}} 
\tableofcontents 




\clearpage 

\setstretch{1.5} 

\lhead{\emph{Abbreviations}} 
\listofsymbols{ll} 
{
\textbf{OTM} & \textbf{O}ne-\textbf{T}ime \textbf{M}emory \\
\textbf{OT} & \textbf{O}blivious \textbf{T}ransfer \\
\textbf{\OT} & One-out-of-Two \textbf{O}blivious \textbf{T}ransfer \\
\textbf{\kOT} & One-out-of-k \textbf{O}blivious \textbf{T}ransfer \\
\textbf{\kROT} & One-out-of-k \textbf{R}andomised \textbf{O}blivious \textbf{T}ransfer \\
\textbf{IQM} & \textbf{I}solated \textbf{Q}ubits \textbf{M}odel \\
\textbf{POVM} & \textbf{P}ositive \textbf{O}perator \textbf{V}alued \textbf{M}easurement \\
\textbf{LOCC} & \textbf{L}ocal \textbf{O}perations and \textbf{C}lassical \textbf{C}ommunication \\

}

\setstretch{1} 

\mainmatter 

\pagestyle{fancy} 


\chapter{Introduction}

\label{Chapter1}

\lhead{Chapter 1. \emph{Introduction}} 

\section{History Of Cryptography: From Art To Science.}

The word cryptography comes from the greek words \greektext{krupt'o} \latintext (``secret") and \greektext{gr'afw} \latintext (``write"). 
In other words it defines the art of secret message transmission between two parties in a way that the message remains unreadable to any third party (\textit{adversary}).
This definition is accurate for the historical uses of cryptography but not for its modern form. 

In the last century, cryptography has evolved from art to science that does not rely on the obscurity of the encryption method but on formal mathematical definitions
and rigorous security proofs. Furthermore, modern cryptography deals not only with the problem of message encryption but also with problems such as authentication, digital signatures and multi-party computation.

In this section we give a brief overview of the history of cryptography and its evolution from the art of message encryption to its modern forms.

\subsection{First Steps: The Art Of Encrypting Messages}
\label{sec:firststeps}
The practice of cryptography is as old as the transmission of messages. Closely linked to the history of mankind, forms of encryption were developed independently in a number
of places and soon again forgotten as were the civilisations that used them.

According to Kahn \cite{Kah96}, cryptography has its roots in 1900 BC ancient Egypt, in the use of unusual hieroglyphs, instead of the ordinary ones, in the tomb of a nobleman, Khnumhotep II.
Together with the construction of impressive burial monuments, the need to impress the living took the form of decorating tombs with obscure encryptions. 
These cryptic puzzles did for the first time
intend to preserve the secrecy of the original text, at least enough to attract the curiosity of passersby for the short time it would take to decrypt and read.

Although there are probably inumerable examples of these first forms of cryptography we note its first known military use to transmit secret messages, the scytale.
First mentioned around the 7th century BC by Apollonius of Rhodes, used by the Spartans the scytale was a method to transmit a message secretly. 
Plutarch gives a more detailed account of its use in \textit{Lives} (Lysander, 19), two identical wooden rods, the scytalae, are used in the following way. A leather strap is wound around the scytale and
then the message is written on it along the length of the rod (see Figure \ref{img:scytale}). The leather strap is then sent to the receiver of the message who has to wind it around his scytale in order to read the message. If the message is intercepted,
it cannot be read unless a rod of the same diameter is used. It is furthermore hypothesized that this could be a method for message authentication instead of encryption, that is only if the sender used the correct 
scytale is the message
 readable by the receiver, thus making it more difficult to inject false messages by a third party.

Through the next centuries, the most common use of cryptography was encryption of text through ciphers by substitution of letters in a fixed way such as the Caesar's cipher. 
The latter uses a fixed left shift of the alphabet by three letters, i.e. A would be transcripted as D, B as E, and so on. 
More complicated ciphers were developped following the same principle, using a, possibly different, shift of the alphabet for every letter of the message, often defined by a secret key.

The most prominent example of complex substitution ciphers is the use of rotor machines, for example the Enigma and Lorentz cipher machines used in World War II (see Figure \ref{img:lorenz}).
These machines used a number of rotating disks (rotors) that implemented a complex, but fixed, substitution of letters. 
For every keypress the position of the rotors would change thus using a different substitution for every letter.

\begin{figure}[H]
 \centering
 \begin{subfigure}{0.45 \textwidth}
  \includegraphics[width= \textwidth, height= 5cm]{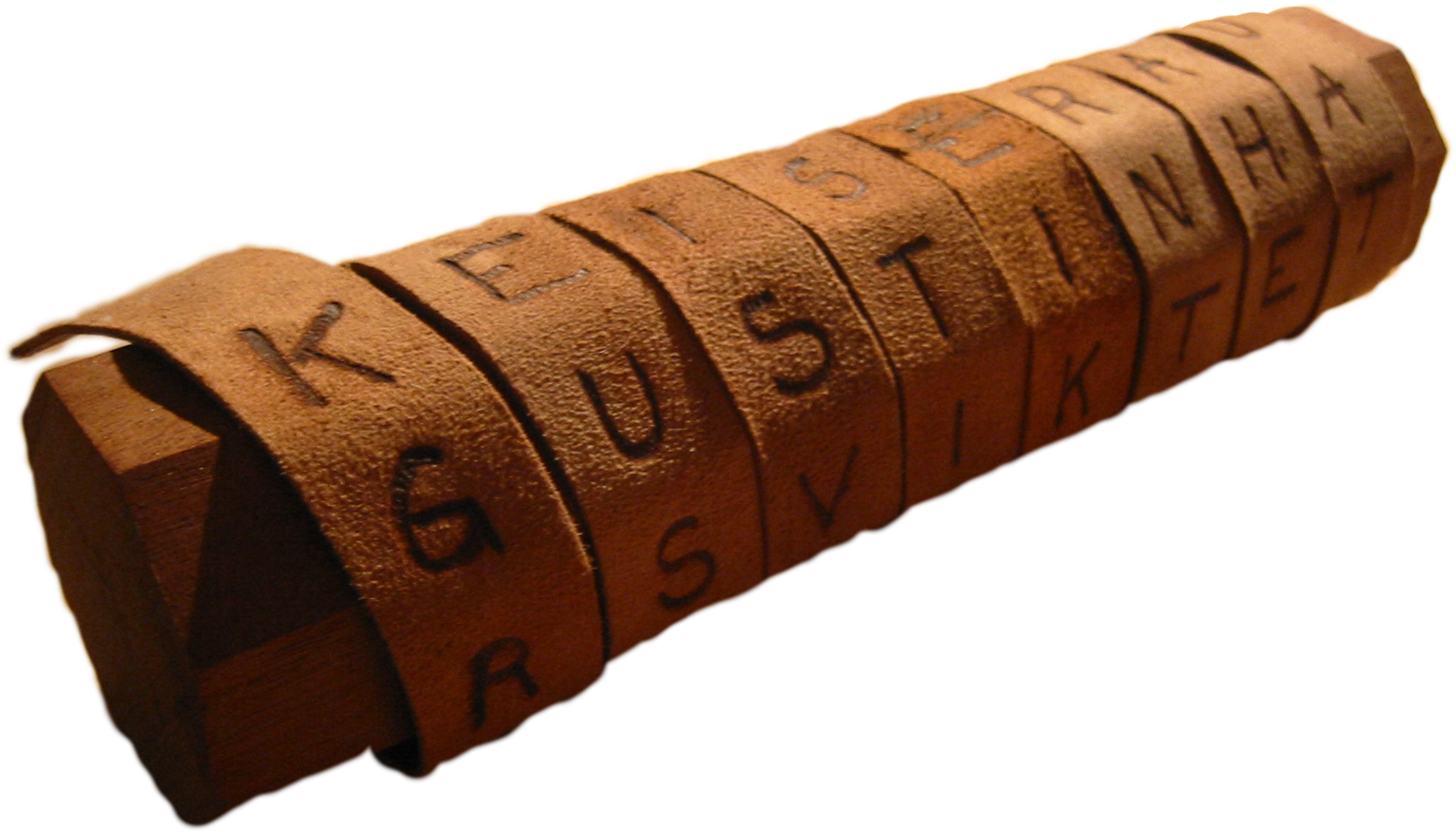}
   \caption{Scytale}
  \label{img:scytale}
 \end{subfigure}~
 \begin{subfigure}{0.45 \textwidth}
  \includegraphics[width= \textwidth, height= 5cm]{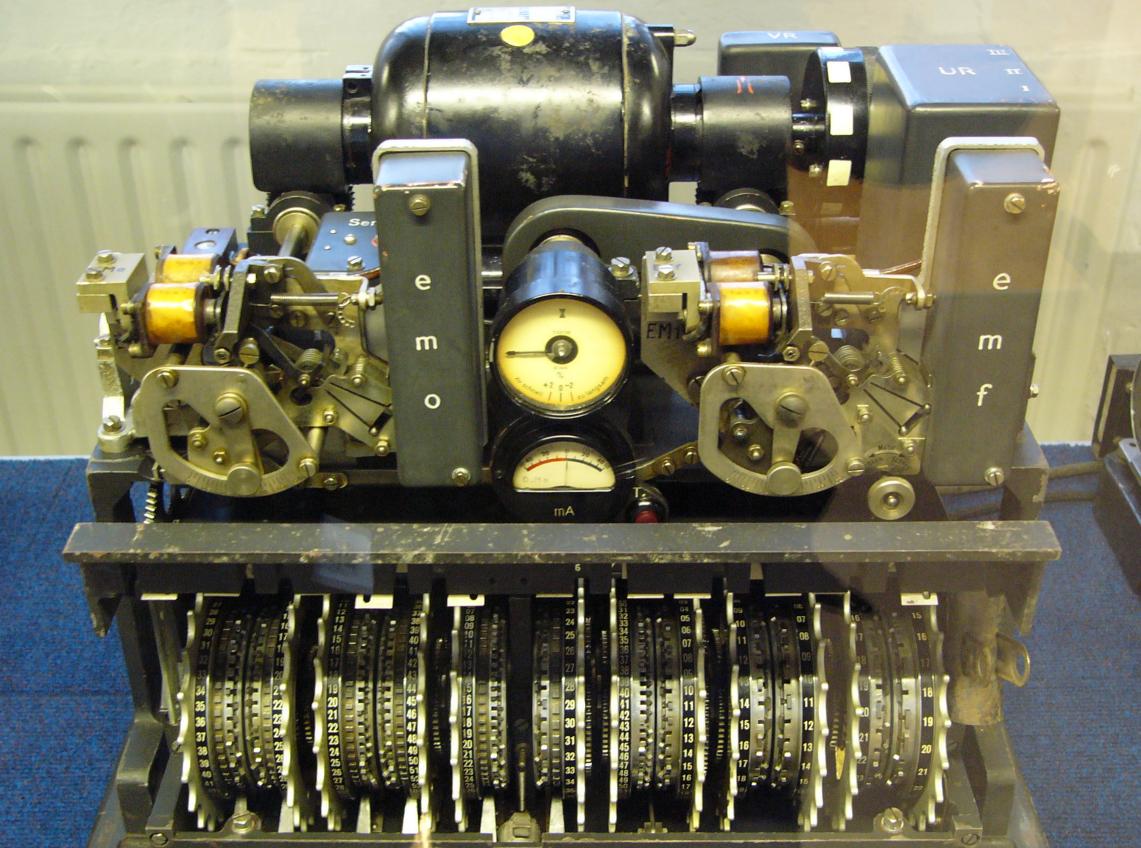}
   \caption{Lorenz rotor stream cipher machine}
  \label{img:lorenz}
 \end{subfigure}
 \caption{Examples of device dependent cryptographic implementations. \fref{img:scytale} The scytale described in more detail in \sref{sec:firststeps} (Source: \url{https://commons.wikimedia.org/wiki/File:Skytale.png})
 and \fref{img:lorenz} The Lorenz SZ42, an example of a rotor cipher machine (Source: \url{https://en.wikipedia.org/wiki/File:Lorenz-SZ42-2.jpg}).}
\end{figure}

\subsection{Modern Cryptography}

For more than twenty centuries cryptography focused mostly on the art of encrypting and conveying secret messages, mainly for military purposes.
A large number of very different and sometimes very complex protocols were implemented, but they all relied on the secrecy of the encryption method. 
Thus once the protocol was known by an adversary it was no longer secure.
The beginning of the end of this era of cryptography was foreseen by A. Kerckhoffs in the following statement:
 \begin{quote}
 A cryptosystem should be secure even if everything about the system, except the key, is public knowledge.
 
 \textit{A.Kerckhoffs}, ``La Cryptographie militaire'', 1883
  \end{quote}

This was later reformulated by C. Shannon as ``the enemy knows the system being used''\cite{Sha49}, starting the modern era of cryptography, where security of cryptographic schemes or protocols
does no longer rely on the obscurity of the encryption methods. For cryptography this was the paradigm shift from art to science.

Modern cryptography relies on the formulation of exact definitions for protocols and rigorous proofs of security \footnote{For a detailed introduction we refer to \cite{KL07}.}. 
Most noteably the security of most 
cryptographic protocols depends on the unproven assumption that some mathematical problems, such as the factorisation of integers \footnotemark, are hard to solve. 
A problem is computationally hard to solve if there exist no algorithms that can do so in polynomial time. This of course means that these protocols are not indefinitely secure since an adversary 
would be able to succeed in violating its security given enough time 
or an algorithm that could solve the problem on which the protocol's security relies efficiently.

Assumptions about the computational restriction of adversaries have so far proved to be sufficiently strong for modern cryptography, but 
 recent developments in quantum computing showed the existence of an algorithm that can factorise
integers in polynomial time if run on a quantum computer, Shor's algorithm \cite{Sho94}. That means that once sufficiently large quantum computers are in use, the implemented cryptographic protocols will become vulnerable.
Faced with this increasingly real danger, cryptographers are trying to develop new approaches to achieve security. 

\footnotetext{Integer factorisation is a widely used computational hardness assumption in cryptographic protocols, for example in RSA \cite{RSA78}. 
So far there exists no algorithm that can solve the problem of factoring a large integer into products of smaller number (usually primes) on a classical computer in polynomial time.}

\subsection{Quantum Cryptography}

In the early 1970's Wiesner proposed the idea of using two-state quantum-mechanical systems, such as polarised photons, to encode and transmit messages \cite{Wie83}. 
Motivated by Heisenberg's uncertainty principle he showed that it is possible to use two ``conjugate observables'', linear and circular polarisation of photons, to ``transmit two messages 
either but not both of which may be received''. This important result remained unpublished for a decade, but set the basis of a new form of cryptography that no longer relies on the 
computational limitation of an adversary to achieve security.
Quantum cryptography is solely based on the assumption that the laws of quantum mechanics model nature accurately to achieve security.
 
Although the first steps of quantum cryptography passed almost unnoticed\footnotemark{}  Brassard and Bennett used Wiesner's idea of ``conjugate coding'' to achieve something previously thought impossible.
The quantum key distribution protocol first developped by Bennett and Brassard and later Ekert \cite{BB84, Eke91, BBE92} that 
allows two users to exchange a secret key over a public quantum communication channel that is eavesdropped on. 
The strength of this quantum protocol lies in the fact that the users are able to detect an eavesdropper who is trying to obtain their key, since measuring a quantum state disturbs its original state.

\footnotetext{For a very enjoyable brief account of these first steps refer to \cite{Bra06}. }

Following this important success in quantum cryptography, the horizons of cryptography broadened and the quest to implement more cryptographic tasks such as secure multi-party quantum computation relying on quantum phenomena to achieve security began. 

Finally it is important to mention post-quantum cryptography as another approach to face the potential threat of quantum computers for the currently implemented cryptographic protocols.
It is the field of search for classical cryptographic assumptions that cannot be broken efficiently by quantum or classical computers \cite{BBD09}.

\section{Secure Two-Party Computation}

We have seen that for a long time cryptographers focused on the problem of transmitting secret messages.
One further problem of cryptography introduced by Yao in \cite{Yao82} is that of secure multi-party computation.
That is the problem where a number of $N$ players each of whom holds an input $x_1,\dots,x_N$ want to evaluate a function of all their inputs, $f(x_1,\dots,x_N)$ correctly without disclosing information
about their respective inputs. This is not only an interesting cryptographic problem, but one that leads to a number of useful applications such as secret voting, oblivious negotiation, private querying of database. 

While Yao introduced the problem of secure multi-party computation, in \cite{Yao82} he mainly focused on the two-party case. That is the problem of
two mutually distrustful parties correctly computing a function without revealing their inputs to each other.

In this thesis we will focus on one problem of two-party computation, namely secure password-based identification:
A user Alice identifies herself to a server Bob by securely evaluating the equality function on their inputs (or passwords).
In the literature this is often refered to as the ``socialist millionaire problem", a variant of the ``millionaire problem''\footnotemark[1],
in which the two millionaires want to determine if they are equally rich, without revealing any information about their actual wealth to each other \cite{Yao82}.

\footnotetext[1]{The millionaire or Yao's millionaire problem is a classic secure multi-party computation problem in which two millionaires want to determine who 
is richer without disclosing any information about their wealth to each other.}

\subsection{Bit Commitment \& Oblivious Transfer}

In this section we focus on two similar but fundamental two-party computation problems, \textit{bit commitment} and \textit{oblivious transfer}, their history and their importance.

Bit commitment schemes consist of two phases, the \emph{commit phase} where the sender Alice chooses the value of a bit and commits to it in the sense that it cannot be changed later and a \emph{reveal phase} during which the hidden value of the bit is revealed and before which the receiver Bob has no information about the value of the bit.

Oblivious transfer is the transfer of information in such a way that the sender does not know what information the receiver obtains. We will give a brief overview of its origin and its importance in secure two-party computation.

The term was coined by Rabin in \cite{Rab81}, where he introduced what is now known as Rabin OT, a protocol where one user Alice sends a message and another user Bob does or does not receive it with equal probability, while Alice
remains oblivious of the reception of the message, this is often refered to as secure erasure channel.

A similar notion was introduced in the first paper on quantum cryptography ``Conjugate Coding'', where Wiesner describes ``a means for transmitting two messages, either but not both of which may be received"  \cite{Wie83}. 
This was later rediscovered by Even, Goldreich and Lempel \cite{EGL85} and named \textit{one-out-of-two oblivious transfer} and denoted as \OT. Intuitively it can be thought of as a black box in which a user Alice can store two messages and another user Bob can choose 
to receive the first or second message but learns no further information about the message he does not receive. Furthermore it fulfills the condition for oblivious transfer, namely that Alice does not know which message Bob received.

A few years later, Cr\'epeau \cite{Cre88} proved that these two flavours of oblivious transfer are equivalent.
In the same year Kilian \cite{Kil88} proved that the \OT primitive is complete for two-party computation. This surprising result meant that a secure \OT construction is sufficient to implement any two-party computation,
making it a fundamental cryptographic problem. Moreover from the results of \cite{Kil88,Cre88} a \OT protocol can be used to implement bit commitment.
 Although a classical protocol was already introduced by Even, Goldreich and Lempel \cite{EGL85} it relies on computational assumptions that are insecure against a quantum adversary.
After the early success of quantum cryptography, 
 research focused on the problem of constructing unconditionally secure bit commitment schemes \cite{BC91,BCJL93} and oblivious transfer or \OT primitives \cite{BBCS92,Cre94}.

Despite these first results, hope to achieve unconditionally secure quantum bit commitment vanished as doing so was proved to be impossible
in a quantum setting in \cite{May96, LC97}. As discussed above, since an \OT primitive can be used to implement bit commitment, the impossiblity result for bit commitment 
implies that \OT is also impossible. In \cite{Lo97}, Lo proved that all quantum one-sided two-party computations, including \OT are insecure.
Furthermore Colbeck in \cite{Col07} and Buhrman et al. in \cite{BCS12} showed that secure two-party computation is impossible to achieve in a fully quantum setting.

One way to circumvent these impossibility results is to impose realistic restrictions on the users.
In the literature there are two successful models that do so, the bounded-quantum-storage model\cite{DFSS07,DFSS08} that upper bounds the size of quantum memory of the users 
and the noisy-storage model\cite{WST08, KWW12, Sch10} that assumes that the quantum memory used is imperfect. Under the assumption of bounding the quantum storage of a user, unconditionally
secure oblivious transfer, \OT and thus two-party computation can be achieved \cite{DFSS07, DFSS08}.

\section{One-Time Memories In The Isolated Qubits Model} \label{intro:OTM}
In 2013 Liu \cite{Liu14a} suggested a further alternative to the memory-restricting models discussed in the previous section, the \textit{isolated qubits model},
where all parties are restricted to the use of
\textit{local operations} on each qubit and \textit{classical communication}(LOCC). The restriction to local quantum operations on each qubit means that the users are not allowed to perform entangling operations on the isolated qubits.
 The model is motivated by experimental work on nitrogen vacancy centers in diamond that can be read out and manipulated optically while at the same time it is difficult to perform entangling operations on pairs of such centers.
 We discuss the isolated qubits model in more detail in \cref{Chapter2}.

 A \textit{one-time memory} (OTM) is a protocol or cryptographic device in which Alice stores two messages and sends it to Bob, who is then able to retrieve only one of the two messages.
 In essence it is a non-interactive or one-way \OT, but we will discuss their difference in more detail in \sref{OTMandOT}
 Liu showed that it is possible to build an imperfect OTM in the IQM that leaks a fraction of information about the unreceived message \cite{Liu14a, Liu14b}. 
 Furthermore, Liu recently showed that it is possible to use privacy amplification in order to achieve a secure OTM for a single bit \cite{Liu15}. 
 
 A significant difference between the isolated qubits model and the noisy- and bounded-quantum-storage models is that the parties are not forced to measure the qubits soon after reception, rather they are allowed to store the qubits for an indefinite amount of time.
 This means that the users are allowed to take advantage of any further information shared between them at a later point to decide on their measurement strategy. 
 On the other hand the noisy- and bounded-quantum-storage models allow entangling operations between the users which is not allowed in the isolated qubits model. 
 In this sense the memory-restricting and isolated qubits models are complementary, which is reflected in the fact that protocols that are secure in one model are not secure in the other. 
 Protocols in the noisy- or bounded-quantum-storage model are insecure in the isolated qubits model in which the adversary has access to unlimited and perfect storage of isolated qubits.
 The opposite is also true since the protocol presented in \cite{Liu14b} is not secure against an adversary that can perform entangling operations.
 We will discuss this in more detail in \sref{model-diff}.

\section{Our Contributions}

 In this thesis, we study the constructions of ``leaky'' string and secure single-bit one-time memories in the isolated qubits model (IQM) introduced in \cite{Liu14a, Liu14b, Liu15}.
 Using non-linear degenerate functions \cite{DFSS06} we simplify the proof presented in \cite{Liu15}. We then construct and prove the security of a string
 \textit{one-out-of-two sender-randomised oblivious transfer}, \ROT, protocol in this model.
 
 Relying on the construction of a secure string \ROT protocol we study for the first time interactive protocols for more complex two-party functionalities in the IQM.
 In order to do so, we assume that all parties measure the qubits they receive at the end of each sub-protocol, which allows us to construct composed protocols.
 First, we construct a \OT protocol that makes use of one instance of a \ROT functionality and prove its security.
 We then construct a weak but efficient \textit{sender-randomised one-out-of-$k$ oblivious transfer}, \sROT, protocol. 
 Finally, we construct a protocol that implements the password-based identification functionality securely, relying on a secure \sROT.

 Moreover, we study the possibility to construct protocols that implement the password-based identification functionality securely and non-interactively.
 We prove that such an implementation is impossible relying only on one-way transmission or even oblivious transfer of messages and qubits from Alice to Bob.
 
 \section{Outline Of The Thesis}

In Chapter \ref{Chapter2}, we introduce notation, the basic concepts from cryptography as well as the model we use in this thesis.
In Chapter \ref{Chapter3}, we extend bit one-time memories introduced in \cite{Liu14a,Liu14b,Liu15} to string \ROT{s} using results from \cite{DFSS06}.
In Chapter \ref{Chapter4}, we study more complex two-party functionalities that make use of multiple instances of the \ROT{s} constructed in Chapter \ref{Chapter3}. 
Firstly, we construct a \OT protocol that makes use of one \ROT functionality.
Secondly, we study  a \kOT protocol presented in \cite{BCR86} that makes use of $k$ \OT{s}. Finally, we present a construction for a weaker but more efficient \sROT protocol that uses only $\log{k}$ \ROT{s}.  
In Chapter \ref{Chapter5}, we prove that constructing a non-interactive identification protocol is impossible even using secure \kOT functionalities.
We then propose a protocol to achieve secure password-based identification and prove its security using the secure \sROT constructed in Chapter \ref{Chapter4}.
In the last Chapter \ref{Chapter6} we summarise our results and discuss their significance.
\chapter{Preliminaries}
\label{Chapter2}
\lhead{Chapter 2. \emph{Preliminaries}}

In this chapter, we introduce notation and the basic tools 
that we will use in this thesis.

We assume some familiarity with basic probability theory and quantum information theory. 
A brief overview of the probability theory notions used in this thesis can be found in Appendix~\ref{AppendixA} and
for an indepth introduction to quantum information theory we refer the reader to \cite{NC00}.

\section{Basic Notation}

We use uppercase letters such as $X,Y,Z$ to denote random variables, calligraphic letters $\Xcal,\Ycal,\mathcal{Z}$ to denote sets and lowercase letters $x,y,z$ to denote a specific value of a random variable.
Furthermore, for a sequence of random variables $X_1,\dots,X_k$ we write $\overline{X_i}$, with $i \in \{1,\dots,k\}$ to denote the sequence $X_1,\dots,X_k$ excluding $X_i$.

Moreover, we introduce the symbol $\pr_{X \lr Y \lr Z}$, as used in \cite{DFSS07} and \cite{FS09}, to denote that the distribution of a random variable $X$ is 
independent of a random variable $Z$ given a random variable $Y$:

\begin{align}
 \pr_{X|YZ} = \pr_{X|Y},
\end{align}
 
 we then write:

\begin{align}
 \pr_{XYZ} = \pr_{X \lr Y \lr Z}.
\end{align}

This notation is extended to $\pr_{XYZ|\cal{E}} = \pr_{X \lr Y \lr Z | \cal{E}}$
to denote that the distribution of a random variable $X$ is independent of a random variable $Z$ given a random variable $Y$ conditioned on an event $\cal{E}$:

\begin{align}
 \pr_{X|YZ\cal{E}} = \pr_{X|Y\cal{E}}.
\end{align}

Finally, the \emph{smoothed min-entropy} of a random variable $X$ conditioned on a random variable $Y$ is denoted by $H^\eps_{\infty}(X|Y)$. 
For more information we refer the reader to Appendix~\ref{AppendixB}.

For any matrix $A \in \CC^{m \times n}$ and vector $x \in \CC^n$ we use $\|A\|,\|A\|_F$ and $\|A\|_{\Tr}$ to denote the operator, the Frobenius and the trace norm respectively. Further information
on these norms is included in Appendix~\ref{AppendixC}.

A brief overview of the Bachmann-Landau symbols:

We write $f(k)=O(g(k))$ if $\exists c>0 ,\,  \exists k_0 , \, \forall k>k_0 : \,  |f(k)|\leq c |g(k)|$.

We write $f(k)=o(g(k))$ if $\forall c>0 ,\, \exists k_0 , \, \forall k>k_0 :\, |f(k)|\leq c |g(k)|$.

We write $f(k)=\Omega(g(k))$ if $\exists c>0 ,\, \exists k_0 , \, \forall k>k_0 :\, |f(k)|\geq c |g(k)|$.

We write $f(k)=\Theta(g(k))$ if $\exists c_1>0 ,\, \exists c_2 >0 ,\, \exists k_0 , \, \forall k>k_0 : \,  c_1 |g(k)| \leq |f(k)| \leq c_2 |g(k)|$.

\section{Functions}

In this section we give a brief overview of special families of functions that we use in this thesis:
 
\subsection{Non-Degenerate Linear Functions}

Non-degenerate linear functions are functions that depend non-trivially on their inputs and are defined in \cite[Definition~4.2]{DFSS06} as follows:

\begin{definition}
\label{def:ndlf}
 A function $\beta: \{0,1\}^\ell \times \{0,1\}^\ell \mapsto \set{0,1}$ is called a non-degenerate linear function if it is of the form:
 \begin{align}
  \beta: (s_0,s_1)\mapsto <u_0,s_0> \xor <u_1,s_1>
 \end{align}

 for two non-zero strings $u_0, u_1 \in \{0,1\}^\ell$, where $<\cdot,\cdot>$ is the bit-wise inner product defined as:
 \begin{align}
  <a,b>=\XOR_{i=1}^{\ell}a_i\cdot b_i
 \end{align}
\end{definition}

We further mention the definition of a more relaxed notion.

\begin{definition}{\cite[Definition~4.3]{DFSS06}}\label{2-balanced1}
A binary function  $\beta: \{0,1\}^\ell \times \{0,1\}^\ell \mapsto \set{0,1}$ is called \emph{2-balanced} if for any $s_0, s_1 \in \{0,1\}^\ell$ the functions
$\beta(s_0,\cdot)$ and $\beta(\cdot,s_1)$ are balanced, meaning that $|\{ \sigma_1 \in \{0,1\}^\ell : \beta(s_0,\sigma_1)=0| = 2^\ell/2 \}$ and
$|\{ \sigma_0 \in \{0,1\}^\ell : \beta(\sigma_0,s_1)=0| = 2^\ell/2 \}$.
\end{definition}

Finally we note the following result that will allow us to use the fact that for any string $s_i$ the functions $\beta(s_i,\cdot)$ and $\beta(\cdot,s_i)$ are balanced in the proof of Lemma~\ref{lem:oneptsec} in \sref{balanced}

\begin{lemma}{\cite[Lemma~4.4]{DFSS06}} \label{2-balanced2}
 Every non-degenerate linear function is 2-balanced.
\end{lemma}

\subsection{$t$-wise Independent Hash Functions}

We introduce the definition of $t$-wise independent hash functions as defined in \cite{Liu15}.

\begin{definition}
 Let $\calH$ be a family of functions $h:\{1,\dots,N\} \mapsto \{1,\dots,M\}$ and $H$ be a function chosen uniformly at random from $\calH$.
 We call $\calH$ a family of $t$-wise independent functions if for all subsets $S \subset \{1,\dots, N\}$ of size $|S| \leq t$, where $t \geq 1$ is an integer,
 the random variables $\{H(x)|x\in S\}$ are independent and uniformly distributed in $\{1,\dots,M\}$.
\end{definition}

Note that sampling and applying a random function from a family of $t$-wise independent hash functions can be done efficiently (\cite[Proposition~2.5]{Liu15}).

We present a large-deviation bound for quadratic functions of $2t-$wise independent random variables \cite[Proposition~2.7]{Liu15}:

\begin{proposition}
 \label{prop:large-bound}
 
 Let $t \geq 2$ be an even integer, and let $\calH$ be a family of $2t$-wise independent functions $\{1,\dots,N\} \mapsto \{0,1\}$.
 Let $A \in \RR^{N\times N}$ be a symmetric matrix, $A^T=A$. Let $H$ be a function chosen uniformly at random from $\calH$, and define the random variable
 \begin{align}
  S= \sum_{x,y=1}^{N} A_{xy} \left( \left(-1\right)^{H(x)}  \left(-1\right)^{H(y)} - \delta_{xy} \right),
 \end{align}

 where $\delta_{xy}$ is the Kronecker $\delta$ that equals 1 if $x=y$ and 0 otherwise.
 
 Then the expected value of $S$ is $\EE[S]=0$ and we have the following large-deviation bound: for any $\lambda > 0$,
 
 \begin{align}
  \pr(|S|\geq \lambda) \leq 4 e^{\frac{1}{6t}} \sqrt{\pi t} \left( \frac{4 \|\Atil \|_F^2 t }{e \lambda^2} \right)^\frac{t}{2} 
  + 4 e^{\frac{1}{12t}} \sqrt{2\pi t} \left( \frac{8 \|\Atil \|^2 t }{e \lambda} \right)^t,
 \end{align}

 where $\Atil \in \RR^{N\times N}$ is the entry-wise absolute value of $A$, that is $\Atil_{xy} = |A_{xy}|$.
\end{proposition}

\section{Functionalities \& Protocols}

An ideal functionality formally describes a cryptographic task, detailing the behaviours of honest and dishonest parties.
A protocol is a series of clearly defined instructions that the (honest) parties follow.
Finally we define the security for a protocol, describing the conditions that need to be fulfilled in order for a protocol to implement a functionality securely.

In this section, we introduce the ideal functionalities of \OT, \ROT, \kOT, \kROT, \sROT and password-based identification as well as equivalent security definitions that we will use in the following chapters.

\subsection{\OT}

First, we formally define the \OT functionality, that we discussed in \cref{Chapter1}, that allows two parties to share one out of two messages such that the sender is oblivious as to which message has been received, while
the receiver has no knowledge of the second message.

\begin{framed}
 \begin{func}\label{func:OT}
 Upon receiving input messages $A_0, A_1 \in \mathcal{X}$ from Alice, where $\mathcal{A}=\{0,1\}^l$ and the choice bit $D \in \{0,1\}$ from Bob, $\F_{\text{\OT}}$
 outputs $A_D$ to Bob and outputs nothing to Alice.
\end{func}
\end{framed}

Commonly security of a protocol is proven by showing that a real protocol is indistinguishable from the ideal functionality. 
However there exists an alternative approach, \cite[Proposition~4.3]{FS09} allows us to use an equivalent security definition. If a protocol fulfills this definition, 
then it securely implements the ideal functionality.

\begin{definition}\label{def:secot}
A $\eps-$secure \OT proocol is a protocol between Alice with inputs $A_0,A_1 \in \Acal$ and Bob with input $D\in\{0,1\}$ such that the following holds:
\begin{description}
  
\item [Correctness:] For honest user Alice and honest server Bob, for any distribution of Alice's inputs  $A_0,A_1 \in \Acal$ and Bob's input $D \in \{0,1\}$,
Alice gets no output and Bob receives output $G=A_D$, except with probability $\eps$.

\item [Security for Alice:] For any dishonest server Bob with output $G'$, there exists a random variable $D' \in \{0,1\}$ such that:

\begin{align}
 \pr_{D'A_0 A_1} \appe \pr_{D'}\cdot \pr_{A_0 A_1}
\end{align}

and 

\begin{align}
 \pr_{G'A_{D'}D'A_{1-D'}} \appe \pr_{G'|A_{D'}D'} \cdot \pr_{A_{D'}D'A_{1-D'}} \label{def-cond1}
\end{align}

\item [Security for Bob:] For any dishonest user Alice with output $V'$,  there exists random variables $A'_0,A'_1$ such that:

\begin{align}
 \pr[G=A'_D] \geq 1 - \eps,
\end{align}

and 
\begin{align}
  \pr_{D V'A'_0 A'_1} \appe \pr_{D} \cdot \pr_{V' A'_0 A'_1}
\end{align}

\end{description}
\end{definition}

\subsection{\ROT}

While \OT is a powerful tool we present a different oblivious transfer functionality, the randomised one-out-of-two oblivious transfer \ROT.
Contrary to the \OT Alice does not input two messages but receives two random messages from the functionality while Bob receives one out of the two messages depending on his input choice.
We present the formal definition of the \ROT functionality.

\begin{framed}
 \begin{func}\label{func:ROT}
 Upon receiving no input from Alice and the choice bit $D \in \{0,1\}$ from Bob, $\F_{\text{\ROT}}$
 outputs messages $A_0, A_1 \in \Acal$, where $\Acal =\{0,1\}^\ell$ to Alice and message $A_D$ to Bob.
\end{func}
\end{framed}

Furthermore, we introduce an equivalent security definition that protocols that securely implement the \ROT functionality should fulfill.

\begin{definition}\label{def:secrot}
A $\eps-$secure \ROT proocol is a protocol between Alice with no input and Bob with input $D\in\{0,1\}$ such that the following holds:
\begin{description}
  
\item [Correctness:] For honest user Alice and honest server Bob, for any distribution of Bob's input $D \in \{0,1\}$,
Alice receives output $A_0,A_1 \in \Acal$ and and Bob receives output $G=A_D$,  except with probability $\eps$.

\item [Security for Alice:] For any dishonest server Bob with output $G$, there exists a random variable $D' \in \{0,1\}$ such that:

\begin{align}
 \pr_{A_{1-D'}GA_{D'}D'} \appe \pr_{U} \cdot \pr_{GA_{D'}D'}
\end{align}

\item [Security for Bob:] For any dishonest user Alice with output $V'$,  there exists random variables $A'_0,A'_1$ such that:

\begin{align}
 \pr[G=A'_D] \geq 1 - \eps,
\end{align}

and 
\begin{align}
  \pr_{D V'A'_0 A'_1} \appe \pr_{D} \cdot \pr_{V' A'_0 A'_1}
\end{align}

\end{description}
\end{definition}

\subsection{\kOT}

In this section, we focus on a generalised oblivious transfer functionality that takes $k$ inputs instead of two, the 1-out-of-k Oblivious Transfer, denoted as \kOT.
It is a two-party functionality between a user Alice that inputs $k$ messages $X_1,X_2,\dots,X_k$ and a user Bob who is allowed to retrieve
only one of these messages $X_D$ according to his choice $D$. 
When the above functionality is implemented securely, Bob should not be able to learn additional information on any of the other messages. At the same time, the obliviousness 
of the protocol must still hold, Alice should not have any knowledge about the choice of Bob.

The formal definition of the \kOT functionality is the following:

\begin{framed}
 \begin{func}\label{func:kOT}
 Upon receiving input messages $X_1, \dots, X_k \in \mathcal{X}$ from Alice, where $\mathcal{X}=\{0,1\}^l$ and the choice $D \in \{1,\dots,k\}$ of Bob, $\F_{\text{\kOT}}$
 outputs $X_D$ to Bob and outputs nothing to Alice.
\end{func}
\end{framed}

We now introduce an equivalent security definition for the \kOT functionality.

\begin{definition}
\label{def:seckot}
A $\eps-$secure \kOT proocol is a protocol between Alice with inputs $X_1,\dots,X_k \in \Xcal$ and Bob with input $D\in\{1,\dots,k\}$ such that the following holds:
\begin{description}
  
\item [Correctness:] For honest user Alice and honest server Bob, for any distribution of Alice's inputs  $X_1,\dots,X_k \in \Xcal$ and Bob's input $D \in \{1,\dots,k\}$,
Alice gets no output and Bob receives output $G=X_D$, except with probability $\eps$.

\item [Security for Alice:] For any dishonest server Bob with output $G'$, there exists a random variable $D' \in \{1,\dots,k\}$ such that:

\begin{align}
 \pr_{D'X_1\dots X_k} \appe \pr_{D'} \cdot \pr_{X_1 \dots X_k}
\end{align}

and

\begin{align}
 \pr_{G'X_{D'}D'\overline{X_{D'}}} \appe \pr_{G'|X_{D'}D'} \cdot \pr_{ X_{D'}D'\overline{X_{D'}}}
\end{align}

\item [Security for Bob:] For any dishonest user Alice with output $V'$,  there exist random variables $X'_1, \dots X'_k$ such that:

\begin{align}
 \pr[G=X'_D] \geq 1 - \eps,
\end{align}

and 
\begin{align}
  \pr_{DV'X'_1 \dots X'_k} \appe \pr_{D} \cdot \pr_{V'X'_1 \dots X'_k}
\end{align}

\end{description}
\end{definition}

\subsection{\kROT}

In this section we introduce a slightly different flavour of the \kOT, where the user Alice does not input messages $X_1,\dots,X_k$ but instead has no inputs and receives as ouptut $k$ random messages $S_1,\dots,S_k$.
This functionality is defined formally below:

\begin{framed}
 \begin{func}\label{func:kROT}
 Honestly behaving Alice and Bob: Upon receiving no input from Alice and a choice  $D \in \{1,\dots,k\}$ from Bob, $\F_{\text{\kROT}}$ samples random independent strings $S_1,\dots,S_k \in \Scal =\{0,1\}^\ell$ and sends $S_1,\dots,S_k$ to Alice and $S_D$ to Bob.
 
 Honest Alice and dishonest Bob: Upon receiving no input from Alice, a choice $D \in \{1,\dots,k\}$ and a string $S_D \in \Scal$ from Bob, $\F_{\text{\kROT}}$ samples random independent strings $\overline{S_D} \in \Scal$, and sends $S_1,\dots,S_k$ to Alice.

 Dishonest Alice and honest Bob: Upon receiving input messages $S_1, \dots, S_k \in \mathcal{S}$ from Alice, where $\mathcal{S}$ and the choice $D \in \{1,\dots,k\}$ of Bob, $\F_{\text{\kROT}}$
 outputs $S_D$ to Bob and outputs nothing to Alice.
\end{func}
\end{framed}

We introduce the security definition for the \kROT functionality.

\begin{definition}\label{def:seckrot}
The sender-randomised \kROT is secure if the following conditions are fulfilled:
\begin{description}
  
\item [Correctness:] For honest user Alice and honest server Bob, for any distribution of Bob's input $D$, Alice gets outputs $S_1,\dots,S_k \in \Scal$ uniform and independent of $D$ and Bob receives output $S_D$, except with probability $\eps$.

\item [Security for Alice:] For any dishonest server Bob with output $G'$, there exists a random variable $D' \in \{1,\dots,k\}$ such that:

\begin{align}
 \pr_{\overline{S_{D'}}S_{D'}D'G'} \appe \pr_{U^{k-1}} \cdot \pr_{S_{D'}D'G'}
\end{align}

\item [Security for Bob:] For any dishonest user Alice with output $V'$,  there exist random variables $S_1',\dots, S_k'$ such that:

\begin{align}
 \pr[G=S_D']\geq 1 - \eps,
\end{align}

and 
\begin{align}
  \pr_{DV'S_1',\dots,S_k'} \appe \pr_{D} \cdot \pr_{V'S_1',\dots,S_k'}
\end{align}

\end{description}
\end{definition}

Finally we introduce the security definition for a slightly weaker \kROT functionality that we call \sROT.

\begin{definition}\label{def:secsrot}
The sender-randomised \sROT is $\eps-$secure if the following conditions are fulfilled:
\begin{description}
  
\item [Correctness:] For honest user Alice and honest server Bob, for any distribution of Bob's input $D$, Alice gets outputs $S_1,\dots,S_k \in \Scal$ uniform and independent of $D$ and Bob receives output $S_D$, except with probability $\eps$.

\item [Security for Alice:] For any dishonest server Bob with output $G'$, there exists a random variable $D' \in \{1,\dots,k\}$ such that for all $I \neq D'$:

\begin{align}
 \pr_{S_{I}S_{D'}D'G'} \appe \pr_{U} \cdot \pr_{S_{D'}D'G'}
\end{align}

\item [Security for Bob:] For any dishonest user Alice with output $V$,  there exist random variables $S_1',\dots, S_k'$ such that:

\begin{align}
 \pr[G=S_D']\geq 1 -\eps,
\end{align}

and 
\begin{align}
  \pr_{DV'S_1',\dots,S_k'} \appe \pr_{D} \cdot \pr_{V'S_1',\dots,S_k'}
\end{align}

\end{description}
\end{definition}

The \sROT is weaker since although every message that does not correspond to Bob's input remains hidden, this is not true for all messages simultaneously.
While weaker, the \sROT functionality is strong enough to construct a secure password-based identification protocol as we will show in \cref{Chapter5}.
Furthermore the \sROT protocol we present in \cref{Chapter4} is more efficient than the \kROT and \kOT protocols, as it makes use of $\log{k}$ instead of $k$ \OT{s}.

\subsection{Password-Based Identification}

  We define the functionality of identification, where a user Alice identifies herself to a server Bob by securely evaluating the equality function on their inputs, called passwords. Our definition is motivated by \cite{FS09}.
  
  \begin{framed}
   \begin{func}\label{func:secid}
 Upon receiving strings $W_{\A} \in \mathcal{W}$ from user Alice, where $\mathcal{W}:=\{1,\dots,k\}$, and $W_{\B} \in \mathcal{W}$ from server Bob, $\F_{ID}$ outputs the bit
 $G=W_{\A} \iseq W_{\B}$ to Bob.
 In case Alice is dishonest she may choose $W_\A = \perp$ (which never agrees with honest Bob's input) and (for any choice of $W_\A$) the bit $G$ is also output to Alice .
  \end{func}
  \end{framed}

  The idea behind the $\F_{ID}$ functionality is that Alice and Bob both have an input string $W_{\A}$ and $W_{\B}$ respectively to act as a password and
  Bob receives and outputs a bit corresponding to the acceptance of Alice's password if their chosen inputs are the same or the rejection if their inputs are
  not equal. In order for a protocol that fulfills the $\F_{ID}$ functionality to be secure, a dishonest server should not be able to learn Alice's password, except with 
  the probability that he guesses the password correctly. At the same time it has to be secure against a dishonest user Alice, so that Bob will not accept her password if
  it does not correspond to his choice $W_\B$. We introduce the definition that should be fulfilled by secure password-based identification protocols.

\begin{definition}\label{def:secid}
A password-based identification protocol is $\eps-$secure if the following conditions are fulfilled:
\begin{description}
  
\item [Correctness:] For honest user Alice and honest server Bob with inputs $W_\A = W_\B$, Bob outputs $G=1$ except with probability $\eps$.

\item [Security for Alice:] For any dishonest server Bob with output $G'$, for any distribution of $W_{\A}$, there exists a random variable
$W'$ that is independent of $W_{\A}$ such that :

\begin{align}
\pr_{W_{\A} W' G' | W'\neq W_\A} \appe \pr_{W_{\A}\lr W' \lr G' | W'\neq W_\A }.
\end{align}

\item [Security for Bob:] For any dishonest user Alice with output $V'$, for any distribution of $W_{\B}$, there exists a random variable $W'$
independent of $W_{\B}$ such that if $W' \neq W_\B$ then $\pr[G=1] \leq \eps $ and :

\begin{align}
\pr_{W_{\B} W' V' | W'\neq W_\B} \appe \pr_{W_{\B} \lr W' \lr V' | W'\neq W_\B}.
\end{align}

\end{description}
\end{definition}

\section{One-Time Memories In The Isolated Qubits Model}

\subsection{The Isolated Qubits Model}

In \cref{Chapter1}, we gave a brief introduction of the \text{isolated qubits model} that was first presented by Liu in \cite{Liu14a}. 
In more detail, parties in this model are restricted to local quantum operations on each qubit and classical communication between the qubits.
As detailed in \cite{Liu14a} any local operation and classical communication (LOCC) strategy, in the sense desribed above, can be described by a series of adaptive single-qubit measurements. 
In subsequent work, Liu describes how to model any LOCC adversary  by a separable positive-operator-value measurement (POVM) \cite{Liu14b}.

Furthermore, in contrast with the memory-restricting models described in \cref{Chapter1}, in the isolated qubits model, all parties are allowed to store qubits for a long time and are not allowed to perform entangling 
operations between qubits. While the restriction on entanglement operations reduces the power of an adversary, the possibility to store qubits for a long time has some important implications.
An adversary is thus allowed to store qubits and measure them at the end of a protocol, making use of any information he receives to decide on his measurement strategy.
Thus usual privacy amplification techniques using hash functions are not effective, which necessitates the use of stronger families of hash functions and a different approach on using them, as described in \cite{Liu15}.
We will describe this in more detail in Section~\ref{privamp}. 

Moreover the ability of storing qubits for a long time allows an adversary to measure the qubits received at the end of the composed protocol\footnotemark.

 It is then not clear if the sub-protocols remains secure.
Composability in the isolated qubits model has not been studied and it seems to be a non-trivial problem. 

\footnotetext{In cryptography, it is common usage to make calls to secure functionalities in a protocol. For example one could use a series of $n$ single-bit commitment functionalities to commit to an $n-$bit string.
One then argues that since every single bit is commited securely, the same holds for the concatenation of these bits. 
Composability of protocols allows one to use a modular design to construct and prove the security of complex protocols. }

In this thesis we assume that all parties have to measure all qubits used in a sub-protocol at the latest at the end of this sub-protocol.
This rather strong assumption allows us to construct composed protocols that make calls to functionalities as sub-routines.

\subsection{Leaky String \ROT}

\label{OTMandOT}

In this section, we introduce a protocol for imperfect \ROT motivated by the ``leaky'' one-time memory (OTM) construction presented in \cite{Liu14b}. 

The security definitions for the ``leaky'' and perfect OTMs presented in \cite{Liu14b,Liu15} are similar to the \ROT security definition, introduced earlier in this chapter.
In \cref{Chapter3} we use the ``leaky'' \ROT presented here to construct protocols a secure string \ROT.
We then use the latter in \cref{Chapter4} to construct a secure \OT protocol.

For consistency with the view of cryptographic tasks as functionalities that are implemented by protocols we do not use the notion of one-time memories as devices that store two messages out of which only one can 
be read. We instead construct protocols that implement the \ROT functionality (Functionality~\ref{func:ROT}) between two users, Alice and Bob. 
The main difference between an OTM and an oblivious transfer protocol is the fact that the first is non-interactive in the sense that only Alice sends information to Bob, while an oblivious transfer protocol
is not necessarily non-interactive.
In that sense, the latter is weaker since an OTM implements the oblivious transfer functionality, but an interactive oblivious transfer protocol does not implement the OTM functionality.

We first rewrite the ``leaky'' OTM as introduced in \cite{Liu14b} as a non-interactive ``leaky'' \ROT protocol that takes no input from Alice and input $D$ from Bob, and outputs $s$ and $t$ to Alice and one of the two messages 
to Bob depending on his input choice. This protocol leaks some information about both messages to Bob and is thus not secure. 

 \begin{prot}
\label{prot:leakyOT}
 A protocol for ``leaky'' string \ROT between users Alice with no input and and Bob with input $D \in \{0,1\}$ respectively.
 
 Let $C':\{0,1\}^\ell \mapsto \{0,1\}^{n\log{q}}$ be an error correcting code that is linear in $GF(2)$ and approaches the capacity of a $q$-ary symmetric channel $\calE_q$ with error probability
  $p_e = \frac{1}{2} - \frac{1}{2q}$.
 \begin{enumerate}
  \item Alice samples and receives as output two strings $s,t \in \{0,1\}^\ell $ uniformly at random.
  \item Alice computes $C'(s)$ and $C'(t)$ and views them as $n$ blocks of $\log{q}$ qubits.
  \item Alice prepares the qubits in the following way and sends them to Bob:
   
   For $i=1,\dots,n$:
   \begin{enumerate}
     \item Let $\gamma_i \in \{0,1\}$ be the outcome of an independent and fair coin toss.
     \item If $\gamma_i = 0$ then prepare the $i^{\text{th}}$ block of $\log{q}$ qubits of the codeword $C'(s)$ in the computational basis: $\ket{C'(s)_i}$
     \item If $\gamma_i = 1$ then prepare the $i^{\text{th}}$ block of $\log{q}$ qubits of the codeword $C'(t)$ in the computational basis: $H^{\otimes \log{q}}\ket{C'(t)_i}$
    \end{enumerate}

    \item Bob measures every qubit in the base corresponding to his input $D \in \{0,1\}$ in the following way:
      \begin{itemize}
       \item If $D=0$, he measures all the qubits he receives in the computational basis.
       \item If $D=1$, he measures all the qubits he receives in the Hadamard basis.
      \end{itemize}

     \item Bob runs the decoding algorithm for $C'$ on the string of measurement outcomes $z \in \{0,1\}^{n \log{q}}$ and receives $s$ or $t$ depending on his choice $D$.
 \end{enumerate}
\end{prot}

We present the definitions for separable measurements and $\delta$-non-negligible measurement outcomes as presented in \cite{Liu14b}, that
are used in Theorem~\ref{leakythm} and later in \cref{Chapter3}.

\paragraph{Separable Measurement}
A measurement on $m$ qubits is called \textit{separable} if it can be written in the form $\calE: \rho \mapsto \sum_i K_i^\dagger \rho K_i$, where each operator
$K_i$ is a tensor product of $m$ single-qubit operators $K_i = K_{i,1} \otimes \dots \otimes K_{i,m}$

\paragraph{$\delta$-non-negligible Measurement Outcome}

\begin{definition}
 For any quantum state $\rho \in \CC^{d\times d}$, and any $\delta >0$, we say that a measurement outcome (POVM element) $M \in \CC^{d \times d}$ is $\delta$-non-negligible if $\Tr(M\rho)\geq \delta \cdot \Tr(M)/d$.
\end{definition}

We rephrase the main result of the original paper, \cite[Theorem~2.3]{Liu14b}, that defines the security of the protocol: 

\begin{theorem}[``Leaky'' String \ROT]
\label{leakythm}
 For any $k\geq2$, and for any small constant $0<\mu<<1$, Protocol~\ref{prot:leakyOT} between Alice with no input and Bob with input $D\in \{0,1\}$, has the following properties:
 \begin{enumerate}
  \item Correctness: For honest users Alice and Bob, Alice receives two messages $s,t \in \{0,1\}^\ell$, where $\ell = \Theta(k^2)$ and Bob receives either $s$ or $t$ depending on his choice $D$, using only LOCC operations.
  \item ``Leaky'' security: Let $\delta_0 > 0$ be any constant, and set $\delta = 2^{-\delta_0 k}$. Honest user Alice receives outputs $s, t \in \{0,1\}^\ell$.
   For any dishonest LOCC Bob,  and any separable measurement outcome $M$ that is $\delta$-non-neglibible, we have the following security bound:
  \begin{align}
     H_\infty^\eps(S,T|Z=M) \geq \left(\frac{1}{2} - \mu \right)\ell - \delta_0k.
  \end{align}
Here $S$ and $T$ are the random variables describing the two messages, $Z$ is the random variable representing the Bob's measurement outcome, and we have $\eps \leq e^{-\Omega(k)}$.
 \end{enumerate}
\end{theorem}

The proof of this theorem can be found in \cite{Liu14b}. This \ROT protocol leaks a constant fraction of information to Bob and is thus not secure for cryptographic tasks. 

\subsection{Privacy Amplification}
\label{privamp}
Common privacy amplification techniques rely on applying a function with a random seed to the string the user holds and require the users to share their seed at a later point.
 These techniques cannot be used in the isolated qubits model as a dishonest user can postpone his measurement until he has knowledge of the seed and use that information to adapt his measurement.

Liu introduces a privacy amplification technique that can be used in the isolated qubits model in \cite{Liu15}. The technique relies on the use of a fixed hash function of a family of $r$-wise hash functions, 
that is a family of stronger hash functions than the ones described above.
This method allows privacy amplification on the output of a leaky string OTM as the ones presented in \cite{Liu14a, Liu14b} and leads to the construction of a secure single-bit OTM \cite{Liu15}.
In \cref{Chapter3} we follow a similar approach to achieve secure string \ROT, instead of the single-bit OTM presented in \cite{Liu15}.

\subsection{Comparing The Isolated Qubits And Bounded Quantum Storage Models}
  \label{model-diff}
  
  In \sref{intro:OTM} we mentioned briefly that the OTM protocols studied in \cite{Liu14a,Liu14b,Liu15} are not necessarily secure in the noisy- and bounded-quantum-storage models and that at the same time
  protocols that rely on a quantum memory bound to achieve security are not guaranteed to be secure in the isolated qubits model. 
  
  In more detail, the OTM protocols constructed in the isolated qubits model are insecure in a model where entangling operations are allowed. 
  An attack against the OTM protocols by an adversary who is allowed to pefrorm entangling operations has been sketched in \cite{Liu14b}, relying on the gentle measurement lemma  \cite{Win99}
  and running the decoding alrgorithm for the error-correcting code on a superposition of many different inputs. This implies that the OTM and \ROT protocols described in \cite{Liu14a,Liu14b,Liu15} and this thesis
  are not secure in the noisy- and bounded-quantum-storage models.
  
  On the other hand protocols in the noisy- and bounded-quantum-storage model \cite{WST08, KWW12, Sch10}, rely on the memory bound or imperfect storage in order to achieve security. In protocols
  such as Protocol~\ref{prot:SchIID}, one user encodes qubits in the computational or Hadamard basis while the receiver measures the qubits either in a random basis or in a sequence of bases depending on his input.
  Since these measurements are destructive, the users commit to a particular choice of measurement bases.
  The correct sequence of bases is announced between the users at a later point, after the memory-bound has been applied. 
  This step allows the users to know which qubits they have measured in the same basis and thus have obtained the same result, unless the quantum communication channel is being eavesdropped on.
  At the same time the step of announcing the bases used to encode the sent qubits can be exploited by a malicious user in the isolated qubits model. Since the users are allowed to store the received qubits for an 
  indefinite amount of time after receiving the qubits, an adversary is allowed to wait until he has received the sequence of bases and thus measure all qubits correctly, which violates the security of these protocols.
  
  Thus we argue that protocols that rely on the restriction of a user to perform non-entangling operations cannot be secure in the memory restricting models. On the other hand protocols that rely on the inability of an adversary to store qubits 
  noiselessly or in large numbers cannot be secure in the isolated qubits model.
   
\chapter{ \ROT In The Isolated Qubits Model }
\label{Chapter3}
\lhead{Chapter 3. \emph{\ROT In The Isolated Qubits Model}}

In this chapter, we introduce a \ROT protocol in the isolated qubits model (IQM), motivated by the ``ideal'' OTM presented in \cite{Liu15}.
Our protocol takes no input from Alice and one bit $D$ as Bob's input, and outputs two strings $A_0$ and $A_1$ to Alice and one string $A_D$ to Bob.
This protocol first uses the ``leaky'' \ROT protocol presented in \cref{Chapter2} and makes use of the privacy amplification technique introduced in \cite{Liu15} to achieve security.
The \ROT protocol differs from the ``ideal'' OTM of \cite{Liu15} in the fact that the messages are strings instead of single bits as in the original.
To prove the security of the \ROT protocol we use some results presented in \cite{DFSS06} that allow us to simplify and extend the proof to longer messages, a technique that was not used in the original.

\section{Secure String \ROT }

As discussed in the previous chapter, the ``leaky'' \ROT, Protocol~\ref{prot:leakyOT}, is not secure because it leaks some information. 
Commonly in such a case one would use privacy amplification techniques to achieve security from this less secure protocol. 
Typically this involves applying a hash function with a seed that is picked by Alice and later announced to Bob, after he has measured the received qubits or messages.

In the isolated qubits model however, the use of such techniques is not possible since Bob is allowed to wait and measure the qubits at a later point, in this case after learning the seed of the 
hash function used for privacy ampification. 
A privacy amplification technique such as this would at best have no effect or even allow a dishonest user Bob to use that information to attack the protocol.
In \cite{Liu15}, Liu presented a technique for privacy amplification in the isolated qubits model by fixing two $r$-wise independent hash functions at the beginning of the protocol, and applying them on the outputs
of the ``leaky'' \ROT protocol. 

\subsection{Protocol String \ROT}

We introduce a protocol for string \ROT based on the protocol proposed by Liu \cite{Liu14b} and the privacy amplification technique that uses two fixed $r$-wise independent hash functions.

\begin{prot}
\label{prot:OTIQM}
  A protocol for string \ROT between user Alice with no input and Bob with input $D \in \{0,1\}$.
\begin{enumerate}
  \item Alice chooses two $r$-wise independent hash functions $F$ and $G$ uniformly at random.
  \item Alice with no input and Bob with input $D$ use a ``leaky'' string \ROT (such as Protocol~\ref{prot:leakyOT}). Alice receives as output two messages $s,t \in \{0,1\}^\ell$ and Bob, depending on his choice, receives $s$ if $D=0$ or $t$ if $D=1$.
  \item Alice receives output  $A_0,A_1 \in \{0,1\}^{\ell'}$ such that:
  \begin{align}
   A_0 &= F(s) \\ A_1 &= G(t)
  \end{align}
  \item Bob computes $F(s)$ or $G(t)$, depending on his input $D$ and obtains $A_D$.
   \end{enumerate}
\end{prot}

\subsection{Security Of The Protocol}

It is not difficult to see that if the ``leaky'' string \ROT is correct then Protocol~\ref{prot:OTIQM} is correct. 
Furthermore since the protocol is non-interactive Alice learns nothing about Bob's actions, as is reasoned in \cite{Liu15}.

The security for Alice of an \ROT, Definition~\ref{def:secrot}, is equivalent to the following definition, that was used in \cite{Liu15}:

\begin{definition}
\label{def-security}
We say that Protocol~\ref{prot:OTIQM} is secure if the following holds:
Let $k\geq 1$ be a security parameter.  
Suppose Alice receives as output two messages $A_0, A_1 \in \set{0,1}$. Consider any dishonest LOCC user Bob, and let $Z$ be the random variable representing the results of Bob's measurements. 
 Then there exists a random variable $D \in \{0,1\}$ such that:
\begin{align}
\begin{split}
\norm{\pr_{A_{1-D} A_{D} D Z} - \pr_{U^{\ell'}} \times \pr_{A_{D} D Z}}_1 \leq 2^{-\Omega(k)},
\end{split}
\end{align}
where $U^{\ell'}$ denotes the uniform distribution on $\set{0,1}^{\ell'}$.
\end{definition}

Theorem~ \ref{thm:main} then states that we can reduce a secure string \ROT protocol (Protocol~\ref{prot:OTIQM}) to a ``leaky'' string \ROT protocol (Protocol~\ref{prot:leakyOT}).
That is if there exists a protocol with output two strings $s,t \in \{0,1\}^\ell$ 
and leaking any constant fraction of information of $s$ and $t$, then we can construct a \ROT where Alice receives two strings $A_0, A_1 \in \{0,1\}^{\ell'}$ and only allows an exponentially small amount of 
information about either $A_0$ or $A_1$ to leak, and is thus secure.

\begin{theorem}
\label{thm:main}
For any constants $\theta \geq 1$, $\delta_0 > 0$, $\alpha > 0$, $\eps_0 > 0$ and $0 < \kappa < \min\left\{\frac{\delta_0}{2}, \frac{\eps_0}{2}, \frac{\alpha}{4}\right\}$  there exists a constant $k_0 \geq 1$ such that:   

\vskip 10pt

Suppose we have a ``leaky'' \ROT protocol in the isolated qubits model, such as Protocol~\ref{prot:leakyOT}, indexed by a security parameter $k \geq k_0$.
  More precisely, suppose that for all $k \geq k_0$, 
\begin{enumerate}
\item Alice receives as output from Protocol~\ref{prot:leakyOT} two messages $s,t \in \set{0,1}^\ell$, where $\ell \geq k$ and uses $m$ qubits, where $k \leq m \leq k^\theta$.
\item Correctness: For honest users, Alice receives $s$ and $t$ and Bob receives $s$ if $D=0$ or $t$ if $D=1$, using only LOCC operations.
\item ``Leaky'' security: Let $\delta_0 > 0$ be any constant, and set $\delta = 2^{-\delta_0 k}$. Honest user Alice receives outputs $s, t \in \{0,1\}^\ell$.
   For any dishonest LOCC Bob, let $Z$ be the random variable representing the result of his measurement. Let $M$ be any separable measurement outcome $M$ that is $\delta$-non-neglibible. Then:
\begin{align}
\label{eq:drf}
H_\infty^\eps(S,T|Z=M) \geq \alpha k,
\end{align}
where $\eps \leq 2^{-\eps_0 k}$.
\end{enumerate}

\vskip 10pt

Now assumbe Alice and Bob use Protocol~\ref{prot:OTIQM}, with $r$-wise independent hash functions $F,G: \{0,1\}^\ell \mapsto \{0,1\}^{\ell'}$, with 
\begin{align}
r = 4(\gamma+1) k^{2\theta}
\end{align}

and 
\begin{align}
 \ell' = \kappa k.
\end{align}

This choice of $r$ is motivated by the union bound, see equation (\ref{eq:ub}).  Here $\gamma$ is some universal constant. The choice of $\ell'$ is motivated by equations \eqref{eq:lreqq1}, \eqref{eq:lreqq2}, \eqref{eq:lreqq3} and \eqref{eq:lreqq4}

\vskip 10pt

Then Protocol~\ref{prot:OTIQM} is a secure \ROT protocol in the isolated qubits model, in the sense of Definition \ref{def-security}. 
 More precisely, for all $k \geq k_0$, the following statements hold, except with probability $e^{-\Omega(k^{2\theta})}$ over the choice of $F$ and $G$:
\begin{enumerate}
\item Alice receives as output from Protocol~\ref{prot:OTIQM} two messages $A_0, A_1 \in \{0,1\}^{\ell'}$ and uses $m$ qubits, where $k \leq m \leq k^\theta$.
\item Correctness: Correctness: For honest users Alice with no input and Bob with input $D$, Alice receives  $A_0$ and $A_1$ and  Bob receives $A_D$, using only LOCC operations.
\item ``Ideal'' security:  For honest Alice with outputs $(A_0, A_1)$ from Protocol~\ref{prot:OTIQM},  
for any dishonest LOCC user Bob, let $Z$ be the random variable representing the results of his measurements.  Then there exists a random variable $D\in \{0,1\}$, such that:
\begin{align}
\begin{split}
&\norm{\pr_{A_{1-D} A_{D} D Z} - \pr_U \times \pr_{A_{D} D Z}}_1 \\
&\qquad \leq 2^{-(\delta_0 k-2(\ell'+1))} + 2^{-(\eps_0k-2\ell'+3)} +  2^{-(\frac{\alpha}{2}k-2(\ell'+1))} 	
 + 2^{-(\frac{\alpha}{2}k-2(\ell'+2+\theta \ln{k})-\ln{(\gamma+1)})} \label{eq:distancebound}\\	
&\qquad\leq 2^{-\Omega(k)},
\end{split}
\end{align}
\end{enumerate}
\end{theorem}

Before proving this theorem we present the definition of the $\eps'-$obliviousness condition  in order to introduce Theorem~ \ref{thm4.5} that we use later to prove the security of Protocol~\ref{prot:OTIQM}.

Note that the $\eps'-$obliviousness condition (for Random 1-2 OT$^\ell$) extended for strings \cite[Definition~3.2]{DFSS06}  describes the security condition of Definition \ref{def-security}.  

\begin{definition}{\textbf{$\eps'$-Obliviousness condition:}}
 For any LOCC adversary who observes the measurement outcome $Z$, there exists a binary random variable $D$ such that 
 \begin{align}
  \| \pr_{A_{1-D}\, A_{D} \, D \, Z}  - \pr_{U^\ell} \times \pr_{A_{D} \ \,D \, Z} \| \leq  \eps'
 \end{align}
\end{definition}

Moreover, we introduce \cite[Theorem~4.5]{DFSS06}, that we will use to prove the security of Protocol~\ref{prot:OTIQM}.

\begin{theorem}
\cite[Theorem~4.5]{DFSS06}\label{thm4.5}
The $\eps'$-obliviousness condition is satisfied for any LOCC adversary who observes the measurement outcome $Z$ if and only if:

 \begin{align}
\forall \text{ non-degenerate linear function $\beta$:}  \| \pr_{\beta(A_0,A_1) \, Z}- \pr_{U^{\ell'}} \times \pr_{Z} \| \leq \frac{\eps'}{2^{2\ell'+1}} \label{secbound2}
 \end{align}
\end{theorem}

Theorem~\ref{thm4.5} states that it is enough to show that  $\| \pr_{(\beta(A_0, A_1))Z}- \pr_U^{\ell'} \times \pr_{Z} \| \leq \frac{\eps'}{2^{2\ell'+1}}$ for all non-degenerate linear functions $\beta$, in order
to prove the security of the protocol.

\section{Proof Of Theorem~\ref{thm:main} }

In this section we  prove Theorem~\ref{thm:main} following the reasoning used in \cite{Liu15}. We  first show that with high probability over $F$ and $G$ the scheme is secure for any fixed separable measurement outcome $M$.
Then we  use the $\mu-$net $\Wtil$ for the set of all separable measurement outcomes and show that Protocol~\ref{prot:OTIQM} is secure at all points $\Mtil \in \Wtil$ with high probability.We  then show that any separable measurement $M$ can be
approximated by a measurement outcome in the $\mu-$net, $\Mtil \in \Wtil$. Then security at $\Mtil$ implies security at $M$ for any separable measurement. Thus Protocol~\ref{prot:OTIQM} is secure. 

\subsection{Security For Fixed Measurement $M$}
First, we  show that in the case when the adversary observes a fixed measurement outcome $Z=M$ the protocol is secure.
Assuming that $M$ is separable and $\delta-$non-negligible, the ``leaky'' security guarantee implies  $H_\infty^\eps(S,T|Z=M) \geq \alpha k$ (equation \eqref{eq:drf}).
The following lemma defines a smoothing event $\calE$ and the quantity $R^{\beta}(M)$ and states that $R^{\beta}(M)$ is small, with high probability over the choice of $F$ and $G$.

\begin{lemma}
\label{lem:oneptsec}
Fix any measurement outcome $M$ such that $H_\infty^\eps(S,T|Z=M) \geq \alpha k$.  
Then there exists an event $\calE$, occurring with probability $P(\calE|Z=M) \geq 1 - \eps$, such that the following statement holds for all non-degenerate linear functions $\beta:\{0,1\}^{\ell'}\times\{0,1\}^{\ell'} \mapsto \{0,1\}$:

We define:
\begin{align}
\label{eqn:bias}
\begin{split}
R^{\beta}(M) &= \EE( 1_\calE \cdot (-1)^{\beta(A_0,A_1)} \,|\, Z=M ),
\end{split}
\end{align}
which is a random variable depending on $F$, $G$, $S$ and $T$, since $A_0=F(S)$ and $A_1=G(T)$.  Then for all $\lambda > 0$ and for all non-degenerate linear functions $\beta$, 
\begin{align}
\label{eq:prbound}
\pr_{F,\, G; \, S,\, T}(\abs{R^{\beta}(M)} \geq \lambda) \leq 8e^{1/(3r)} \sqrt{\pi r} \biggl( \frac{8\cdot 2^{-\alpha k} r^2}{e^2 \lambda^2} \biggr)^{r/4}.
\end{align}
\end{lemma}

\begin{proof}
From $H_\infty^\eps(S,T|Z=M) \geq \alpha k$, there exists a smoothing event $\calE$, occurring with probability $P(\calE \,|\, Z=M) \geq 1-\eps$, such that:
\begin{align} \forall \; s,t \in \set{0,1}^\ell, \;
P(\calE,\, S=s, T=t|\, Z=M) \leq 2^{-\alpha k}.
\end{align}

Then the following holds: 

\begin{align}\begin{split}
 \sum_{s,t \in \set{0,1}^\ell}P(\calE, S=s& , T=t \,|\, Z=M)^2 \\
 &= \sum_{s,t \in \set{0,1}^\ell} P(\calE, S=s , T=t \,|\, Z=M) \cdot P(\calE, S=s , T=t \,|\, Z=M) \\
 & \leq \sum_{s,t \in \set{0,1}^\ell} 2^{-\alpha k} \cdot P(\calE, S=s , T=t \,|\, Z=M) \\
 & = 2^{-\alpha k} \cdot  \sum_{s,t \in \set{0,1}^\ell} P(\calE, S=s , T=t \,|\, Z=M) \\
 & \leq 2^{-\alpha k} \label{eq:pb2}
\end{split}
\end{align}

We  now bound the quantity $R^{\beta}(M)$ in a similar way as in \cite{Liu15}.
For a non-degenerate linear function $\beta$ defined by non-zero strings $u_0,u_1$, $\beta(A_0,A_1) = <u_0,F(s)>+<u_1,G(t)>$, where by definition $A_0 = F(s)$ and $A_1 = G(t)$.
We then rewrite $R^{\beta}(M)$ as

\begin{align}
R^{\beta}(M) = \sum_{s,t \in \set{0,1}^\ell} (-1)^{<u_0,F(s)>+<u_1,G(t)>} P(\calE, S=s, T=t \,|\, Z=M).
\end{align}

Firstly, we define a function $H:\: \set{0,1} \times \set{0,1}^\ell \rightarrow \set{0,1}$: 
\begin{align}
H(i,s) = \begin{cases}
<u_0,F(s)>, &\text{if $i=0$} \\
<u_1,G(s)>, &\text{if $i=1$},
\end{cases} \label{eq:funA}
\end{align}
 
for two non-zero $u_0,u_1 \in \{0,1\}^{\ell'}$. 

Note that since $F,G$ are $r-$wise independent hash functions and $u_0,u_1$ are non-zero strings then $<u_0,F(s)>$ and $<u_1,G(s)>$ are also $r-$wise independent hash functions.

By definition, F is a $r-$wise independent hash function if for all subsets $S \subset \{0,1\}^{\ell}$ of size $|S| \leq r$, the random variables $\{F(x)|x \in S\}$ are independent and uniformly distributed in $\{0,1\}^{\ell'}$.
Then the random variables $\{<u_0,F(x)>|x\in S\}$ are also independent, where $<u_0,F(x)>=\XOR_{i=1}^{\ell'}u_{0_i} \cdot \{F(x)\}_i$.
Furthermore, from the fact that all non-degenerate linear functions are $2-$balanced, Lemma~\ref{2-balanced2}, and from the definition of $2-$balanced functions, Definition~\ref{2-balanced1}, we can see that  since $u_0$ is non-zero the random variables  $\{<u_0,F(x)>|x\in S\}$ are uniformly distributed.
(The same holds for  $\{<u_1,G(x)>~|~x~\in ~S\}$). \label{balanced}

Secondly, we define a matrix $A \in \RR^{(2\cdot 2^\ell) \times (2\cdot 2^\ell)}$ with entries $A_{(i,s)(j,t)}$, for $ i,j \in \set{0,1} $ and $ s,t \in \set{0,1}^\ell$, that take the values: 
\begin{align}
A_{(i,s),(j,t)} = 
\begin{cases}
\tfrac{1}{2} P(\calE, S=s, T=t \,|\, Z=M) &\text{if $(i,j)=(0,1)$} \\
\tfrac{1}{2} P(\calE, S=t, T=s \,|\, Z=M) &\text{if $(i,j)=(1,0)$} \\
0 &\text{otherwise}.
\end{cases} \label{eq:matA}
\end{align}

Finally, using \eref{eq:funA} and \eref{eq:matA}, $R^{\beta}(M)$ can be written in the following way,

\begin{align}
R^{\beta}(M) &= \EE( 1_\calE \cdot (-1)^{\beta(A_0,A_1)} \,|\, Z=M ) \\
&=\sum_{s,t \in \set{0,1}^\ell} P(\calE, S=s, T=t \,|\, Z=M) (-1)^{<u_0,F(s)>+<u_1,G(t)>} \\
&= \sum_{s,t \in \set{0,1}^\ell} \Bigl\{ \tfrac{1}{2} P(\calE, S=s, T=t \,|\, Z=M) (-1)^{<u_0,F(s)>+<u_1,G(t)>} \\
& \; \; \; \; \; \;\;\;\;\;\;\;\;\;\;\;\;\;  +  \tfrac{1}{2} P(\calE, S=t, T=s \,|\, Z=M) (-1)^{<u_1,G(t)>+<u_0,F(s)>} \Bigr\} \\
& = \sum_{(i,s),(j,t)} A_{(i,s),(j,t)} \bigl( (-1)^{H(i,s)} (-1)^{H(j,t)} - \delta_{(i,s),(j,t)} \bigr)
\end{align}

Since $<u_0,F>$ and $<u_1,G>$ are $r$-wise independent random functions, we can set $t = r/2$ and use Proposition~\ref{prop:large-bound},
using the following bounds on $\Atil$:
\begin{align}
\begin{split}
\norm{\Atil}^2 \leq \norm{\Atil}_F^2 &= \sum_{(i,s),(j,t)} A_{(i,s),(j,t)}^2 \\
&= \tfrac{1}{2} \sum_{s,t} P(\calE, S=s, T=t \,|\,  Z=M)^2 \\
&\leq \tfrac{1}{2} \cdot 2^{-\alpha k}, 
\end{split}
\end{align}
where in the last line we used equation (\ref{eq:pb2}).  
Then by substituting into Proposition~\ref{prop:large-bound} we prove equation (\ref{eq:prbound}).
We thus prove Lemma~\ref{lem:oneptsec}.
\end{proof}

Next, we introduce Lemma~\ref{lem:Rxor} that implies that if $R^{\beta}(M)$ is small, we can use Theorem~ \ref{thm4.5} to prove the security of the protocol when the adversary observes the measurement outcome $M$.

\begin{lemma}
\label{lem:Rxor}
 Fix any measurement outcome $M$. Suppose $|R^{\beta}(M)|\leq \xi$. Then:
 
 \begin{align}
  \|\pr_{\beta(A_0,A_1), \calE|Z=M} - \pr_U \| \leq \xi + \eps
 \end{align}
\end{lemma}

\begin{proof}
 
 Fix a measurement outcome $M$ and suppose $|R^{\beta}(M)| \leq \xi$.
From the definition of $R^{\beta}(M)$ we have that:

\begin{align}
 R^{\beta}(M) &=\EE( 1_\calE \cdot (-1)^{\beta(A_0,A_1)} \,|\, Z=M ) \\
  &= P(\beta(A_0,A_1)=0, \calE | Z=M) -P( \beta(A_0,A_1) =1 , \calE|Z=M)
\end{align}

From $R^{\beta}(M)\leq \xi$ :
\begin{align}
   - \xi \leq P(\beta(A_0,A_1)=0, \calE | Z=M) -P(\beta(A_0,A_1) =1 , \calE|Z=M) \leq \xi \label{eq:lemRxorineq}
\end{align}

From $P(\calE|Z=M)\geq 1-\eps$ and basic probability theory:

\begin{align}
  1-\eps \leq P(\calE|Z=M) = &P(\beta(A_0,A_1)=0, \calE | Z=M) + P( \beta(A_0,A_1) =1 , \calE|Z=M)  \leq 1 \label{eq:ineq1} 
\end{align}

Combining equation \eqref{eq:lemRxorineq} with equation \eqref{eq:ineq1}
we get:

\begin{align}
 \left|  P(\beta(A_0,A_1)=0, \calE | Z=M) - \frac{1}{2} \right| \leq \frac{\xi+\eps}{2}
\end{align}

and 

\begin{align}
   \left| P(\beta(A_0,A_1)=1, \calE | Z=M) - \frac{1}{2} \right| \leq \frac{\xi+\eps}{2} \label{eq:pt2}
\end{align}

Then the $\ell_1$ distance between $\pr_{\beta(A_0,A_1), \calE|Z=M}$ and $\pr_U$ is:

\begin{align}
 \begin{split}
  \|\pr_{\beta(A_0,A_1),\calE|Z=M}-\pr_U\| &= \left|P(\beta(A_0,A_1)=0,\calE | Z=M) - \frac{1}{2} \right|
  \\& \, \,\, \, \,+  \left|P(\beta(A_0,A_1)=1,\calE | Z=M) - \frac{1}{2} \right|\\ &\leq \xi + \eps
 \end{split}
\end{align}
\end{proof}

Thus we have proven that if $R^{\beta}(M)$ is small, Lemma~\ref{lem:Rxor} together with Theorem~\ref{thm4.5} imply that
 Protocol~\ref{prot:OTIQM} is secure againsta dishonest user Bob that observes the measurement outcome $M$.

\subsection{Security For $\mu-$net}

In \cite{Liu15}, it is shown that there exists an $\mu-$net $\Wtil$ for the set of all possible separable measurement outcomes $W$ with respect to the operator norm $\|\cdot\|$.
In this section, we show that the protocol is secure for all the measurement outcomes in the $\mu-$net.

First, we introduce the following lemma as presented and proved in \cite{Liu15}.

\begin{lemma}{\cite[Lemma~3.5]{Liu15}} \label{lemma:3.5}
 For any $0 < \mu \leq 1$, there exists a set $\Wtil \subset W$, of cardinality $|\Wtil| \leq \left(\frac{9m}{\mu}\right)^{4m}$, which is a $\mu-$net for $W$ with respect to the operator norm $\|\cdot\|$.
\end{lemma}

We then use Lemma~\ref{lemma:3.5}, and set
  
\begin{align}
 \mu = 2^{-(\alpha/2)k} \cdot \frac{\delta^2}{4^m},
\end{align}

The value of $\mu$ is chosen so that it is small enough to approximate any measurement outcome, see equation \eqref{eq:37} in the next section.

Together with the fact that $k \leq m \leq k^{\theta}$ and $\delta = 2^{-\delta_0 k}$ the cardinality of $\Wtil$ is bounded by:

\begin{align}
 |\Wtil| & \leq \left( 9m \cdot 2^{\frac{\alpha}{2}k}\frac{4^m}{\delta^2} \right)^{4m} = \left(2^{\log(9m)+\frac{\alpha}{2}k+2\delta_0 k +2m} \right)^{4m}  \\ 
 & = 2^{4m\log(9m)+ 4(\alpha/2+2\delta_0) km +8m^2} \leq 2^{4m\log(9m)+ (2\alpha+8\delta_0 + 8 )m^2} \leq 2^{4k^\theta\log(9k^\theta)+ (2\alpha+8\delta_0 + 8 )k^{2\theta}} \\ 
 & =  2^{4k^\theta\log{9}+4 k^\theta \theta\log{k}+ (2\alpha+8\delta_0 + 8 )k^{2\theta}} .
\end{align}

For sufficiently large $k$ it holds that $\log{k}\leq k \leq k^\theta \leq k^{2 \theta}$. This also implies that $k^\theta \log{k} \leq k^{2\theta}$. Then for all sufficiently large $k$,

\begin{align}
 |\Wtil| & \leq 2^{4k^\theta\log{9}+4 k^\theta \theta\log{k}+ (2\alpha+8\delta_0 + 8 )k^{2\theta}}  \leq 2^{(4\log{9}+4\theta+ 2\alpha+8\delta_0 + 8 )k^{2\theta}} \\
 &\leq 2^{\gamma k^{2 \theta}},
 \end{align}

where $\gamma$ is a constant.

Next we use Lemma~\ref{lem:oneptsec} and we set 

\begin{align}
 \lambda = 2^{-(\alpha/2)k} \cdot 2r.
\end{align}

Then we have that

\begin{align}
  \pr_{F,\, G; \, S,\, T}(|R^{\beta}(M)|\geq \lambda ) \leq 8 e^{1/3r} \sqrt{\pi r} (e^2 /2)^{-r/4}.
\end{align}

Finally, using the union bound we show that with high probability for all $\Mtil \in \Wtil$ and all non-degenerate linear functions $\beta$, $R^{\beta}(\Mtil)$ is small.

\begin{align}
\pr_{F,\, G; \, S,\, T}\Bigl( &\exists \Mtil \in \Wtil,\, 
\text{s.t. $\Mtil$ is $\delta$-non-negligible, and } \exists \beta \text{ s.t. } 
\abs{R^{\beta}(\Mtil)}  \geq \lambda \Bigr) \\
&\leq  |\Wtil| \cdot \sum_{\beta} \pr_{F,\, G; \, S,\, T}(|R^{\beta}(\Mtil)|\geq \lambda) \\
&\leq  |\Wtil| \cdot 2^{2\ell'} \cdot \pr_{F,\, G; \, S,\, T}(|R^{\beta}(\Mtil)|\geq \lambda) \\
&\leq  2^{\gamma k^{2\theta}} \cdot 2^{2\ell'}\cdot 
\Bigl( 8e^{1/(3r)} \sqrt{\pi r} (e^2/2)^{-r/4} \Bigr) \\
&\stackrel{r=4(\gamma+1)k^{2\theta}}{=} 2^{\gamma k^{2\theta}+3+(\gamma+1)k^{2\theta} + 2 \ell'} \cdot 
 e^{\frac{1}{12(\gamma+1)k^{2\theta}} + \frac{1}{2} \ln{4\pi (\gamma+1)k^{2\theta}}+2(\gamma+1)k^{2\theta}}  \\
&\stackrel{\ell'=\kappa k}{=}   \exp \left\{ (3+(2\gamma+1)k^{2\theta} + 2 \kappa k )\ln{2} \right. \\
&\, \, \, \,  \left. +\frac{1}{12(\gamma+1)k^{2\theta}} + \frac{1}{2} \ln{4\pi (\gamma+1)}+\theta \ln{k}-2(\gamma+1)k^{2\theta}\right\}  \label{eq:346}
\end{align}

Since $k^{2\theta}\ln{2}>0$ and $e^{k^{2\theta}\ln{2}} \geq 1$  we multiply \eref{eq:346} with $e^{k^{2\theta}\ln{2}}$.

Furthermore, using the fact that 
\begin{align} 
          f(k)= 2\kappa k \ln{2} + \theta \ln{k}+\frac{1}{12(\gamma+1)k^{2\theta}} + 3 \ln{2}+ \frac{1}{2} \ln{4\pi (\gamma+1)} = o(k^{2}),
\end{align}

since 

\begin{align}
 \lim_{k\rightarrow \infty} \frac{f(k)}{k^{2}} = \lim_{k\rightarrow \infty} \frac{2\kappa k \ln{2} \theta \ln{k}+\frac{1}{12(\gamma+1)k^{2\theta}} + 3 \ln{2}+ \frac{1}{2} \ln{4\pi (\gamma+1)}}{k^{2}} = 0,
\end{align}

\eref{eq:346} becomes:

\begin{align}
 &\exp\left\{(2\gamma+1)k^{2\theta}\ln{2}-2(\gamma+1)k^{2\theta}+2 \kappa k \ln{2} + \theta \ln{k}\right. \\ 
&\, \, \, \, \left. +\frac{1}{12(\gamma+1)k^{2\theta}} + 3 \ln{2}+ \frac{1}{2} \ln{4\pi (\gamma+1)}\right\} \\
&\leq   \exp{\left\{2(\gamma+1)(\ln{2}-1)k^{2\theta}+o(k^2)\right\}}   \\
&= \exp{\left\{- \left( 2(\gamma+1)(1-\ln{2})k^{2\theta} - o(k^2) \right)\right\}}.
\end{align}

Thus 
\begin{align}
\pr_{F,\, G; \, S,\, T}\Bigl( &\exists \Mtil \in \Wtil,\, 
\text{s.t. $\Mtil$ is $\delta$-non-negligible, and } \exists \beta \text{ s.t. } 
\abs{R^{\beta}(\Mtil)}  \geq \lambda \Bigr) \leq e^{-\Omega(k^{2\theta})}. \label{eq:ub}
\end{align}

Equation \eqref{eq:ub} implies that with high probability over $F$ and $G$,

\begin{align}
\label{eq:provfg}
\forall \Mtil \in \Wtil , \,  \text{($\Mtil$ is $\delta$-non-negligible)} 
\Rightarrow  \abs{R^{\beta}(\Mtil)} \leq \lambda.
\end{align}

Thus Protocol~\ref{prot:OTIQM} is secure in the case where the adversary observes any measurement outcome in the set $\Wtil$.

\subsection{Approximating Measurement Outcomes}

We now show that any measurement outcome $M$ can be approximated by another measurement outcome $\Mtil$, just as in \cite{Liu15}.

First, we introduce a lemma proved and used in \cite{Liu15} that shows that if $M$ is $2\delta-$non-negligible then $\Mtil$ is $\delta-$non-negligible.

\begin{lemma}{\cite[Lemma~3.6]{Liu15} } \label{lem:3.6}
Suppose that $M,\Mtil \in \left(\CC^{2\times2} \right)^{\otimes m}$, and $0 \preceq M \preceq I$, and $0 \preceq \Mtil \preceq I$. Suppose that $M$ is $2\delta-$non-negligible, where $0 < \delta \leq \frac{1}{2}$, and 
$\Tr(M) \geq 1$. Suppose that $\Mtil$ satisfies $\| M - \Mtil \| \leq \mu$, where $\mu \leq \frac{2}{3} \delta \cdot 2^{- m}$. Then $\Mtil$ is $\delta-$non-negligible.
\end{lemma}

The next lemma shows that if the quantity $R^{\beta}(\Mtil)$ is defined in terms of an event $\calEtil$, we can define the quantity $R^{\beta}(M)$ by choosing $\calE$ such that $R^{\beta}(M) \approx R^{\beta}(\Mtil)$.

\begin{lemma}
\label{lem:continuity}
Suppose that $M, \Mtil \in (\CC^{2\times 2})^{\tensor m}$, and $0 \preceq M \preceq I$, and $0 \preceq \Mtil \preceq I$.  Suppose that $M$ is $2\delta$-non-negligible, where $0 < \delta \leq \tfrac{1}{2}$, and $\norm{M} = 1$.  Suppose that $\Mtil$ satisfies $\norm{M-\Mtil} \leq \mu$, where $\mu \leq \tfrac{1}{2}$, and $\Mtil$ is $\delta$-non-negligible.

Suppose there exists an event $\calEtil$, occurring with probability $P(\calEtil|\Ztil=\Mtil)$; and let  $R^{\beta}(\Mtil)$ be defined in terms of $\calEtil$, as shown in equation (\ref{eqn:bias}).

Then there exists an event $\calE$, occurring with probability $P(\calE|Z=M) = P(\calEtil|\Ztil=\Mtil) $, such that if  $R^{\beta}(M)$ is defined in terms of $\calE$, then the following statement holds:

\begin{align}
\abs{R^{\beta}(M) - R^{\beta}(\Mtil)} \leq 2\mu \biggl( \frac{2^m}{\delta} \biggr)^2.
\end{align}
\end{lemma}

\begin{proof}

By assumption, there is an event $\calEtil$, defined by the probabilities $P(\calEtil \;|\; \Ztil=\Mtil,\, S=s,\, T=t)$. 
Let $\rho_{st}$ be the quantum state used to encode messages $(s,t)$, i.e., 
this is the state of the one-time memory, conditioned on $S=s$ and $T=t$.  

We start by writing $R^{\beta}(\Mtil)$ in a more explicit form.  First consider $R^{\beta}(\Mtil)$, and note that $A_0 = F(S)$, $A_1 = G(T)$ and $s,t$ are chosen uniformly at random.  We can write $R^{\beta}(\Mtil)$ in the form: 
\begin{align}
\begin{split}
R^{\beta}(\Mtil) &= \frac{1}{P(\Ztil=\Mtil)} \sum_{s,t \in \set{0,1}^\ell} (-1)^{\beta(F(s),G(t))} 
P(\calEtil,\, S=s,\, T=t,\,\Ztil=\Mtil) \\
 &= \frac{1}{P(\Ztil=\Mtil)} \sum_{s,t \in \set{0,1}^\ell} (-1)^{\beta(F(s),G(t))}  
P(\calEtil,\, |\, \Ztil=\Mtil,\, S=s,\, T=t) \Tr(\Mtil\rho_{st}) \, 4^{-\ell} \\
 &= \frac{1}{P(\Ztil=\Mtil)} \, \Tr(\Mtil\nu),
\end{split} \label{eq:355}
\end{align}
where we define the matrix $\nu \in (\CC^{2\times 2})^{\tensor m}$ as follows:
\begin{align}
\begin{split}
\nu = 4^{-\ell} \sum_{s,t \in \set{0,1}^\ell} (-1)^{\beta(F(s),G(t))} 
 &P(\calEtil |\, \Ztil=\Mtil,\, S=s,\, T=t) \, \rho_{st}.
\end{split}
\end{align}

The trace norm of $\nu$ is $\norm{\nu}_{\Tr}=\Tr(\sqrt{\nu \nu^\dagger})$. But note that

\begin{align}
\nu^\dagger &= \left( 4^{-l} \sum_{s,t} (-1)^{\beta(F(s),G(t))} P(\calEtil|\Ztil=\Mtil,\,S=s,\,T=t)\,\rho_{st}\right)^\dagger \\
&= 4^{-l} \sum_{s,t} (-1)^{\beta(F(s),G(t))} P(\calEtil|\Ztil=\Mtil,\,S=s,\,T=t)\,\rho_{st}^\dagger
 \end{align}

Since density operators are self-adjoint $\rho_{st}^\dagger=\rho_{st}$. Then 

\begin{align}
\nu^\dagger = 4^{-l} \sum_{s,t} (-1)^{\beta(F(s),G(t))} P(\calEtil|\Ztil=\Mtil,\,S=s,\,T=t)\,\rho_{st} = \nu. 
\end{align}

Then the trace norm is: 
\begin{align} \norm{\nu}_{\Tr} = \Tr(  4^{-l} \sum_{s,t} (-1)^{\beta(F(s),G(t))} P(\calEtil|\Ztil=\Mtil,\,S=s,\,T=t)\,\rho_{st} ) \\
=~4^{-l} \sum_{s,t} (-1)^{\beta(F(s),G(t))} P(\calEtil|\Ztil=\Mtil,\,
S=s,\,T=t)\, \Tr(\rho_{st}) 
\end{align}

Since $\rho_{st}$ is a density operator  $\Tr(\rho_{st}=1)$
and because 
\begin{align} P(\calEtil|\Ztil=\Mtil,\,S=s,\,T=t) \leq  1 \\ \Rightarrow  (-1)^{\beta(F(s),G(t))} P(\calEtil|\Ztil=\Mtil,\,S=s,\,T=t) \leq 1 \\
\Rightarrow \sum_{s,t}  (-1)^{\beta(F(s),G(t))} P(\calEtil|\Ztil=\Mtil,\,S=s,\,T=t) \leq 4^l
\end{align}
we get:

\begin{align} 
\norm{\nu}_{\Tr} = 4^{-l} \sum_{s,t} (-1)^{\beta(F(s),G(t))} P(\calEtil|\Ztil=\Mtil,\,S=s,\,T=t) \leq 4^{-l}\cdot 4^l.
\end{align}

Thus $\norm{\nu}_{\Tr} \leq 1$.

In addition, we can rewrite $P(\Ztil=\Mtil)$ in the following way:

\begin{align}
\begin{split}
P(\Ztil=\Mtil) &= \sum_{s,t \in \set{0,1}^\ell} P(\Ztil=\Mtil,\, S=s,\, T=t) \\
 &= \sum_{s,t \in \set{0,1}^\ell} \Tr(\Mtil\rho_{st}) \, 4^{-\ell} \\
 &= \Tr\left(\Mtil \left( 4^{-\ell} \sum_{s,t \in \set{0,1}^\ell} \rho_{st}\right) \right) \, \\
 &= \Tr(\Mtil\xi), 
\end{split} \label{eq:357}
\end{align}
where we define the matrix $\xi \in (\CC^{2\times 2})^{\tensor m}$ as follows:
\begin{align}
\xi = 4^{-\ell} \sum_{s,t \in \set{0,1}^\ell} \rho_{st}.
\end{align}
Also, note that $\norm{\xi}_{\Tr} \leq 1$.

Taking into account \eref{eq:355} and \eref{eq:357} we can rewrite $R^\beta(\Mtil)$ as
\begin{align}
\label{eq:mtilbias}
R^{\beta}(\Mtil) = \frac{\Tr(\Mtil\nu)}{\Tr(\Mtil\xi)}, 
\end{align}
where $\nu, \xi \in (\CC^{2\times 2})^{\tensor m}$ satisfy $\norm{\nu}_{\Tr} \leq 1$ and $\norm{\xi}_{\Tr} \leq 1$.  

We now consider the measurement outcome $M$.  We  construct 
an event $\calE$ in order to define the quantity $R^{\beta}(M)$. The event $\calE$ conditioned on $Z=M$  behaves similarly to $\calEtil$ conditioned on $\Ztil=\Mtil$.

We now construct the event $\calE$. For a fixed measurement outcome $M$ and for all $s,t \in \set{0,1}^\ell$ we define:
\begin{align}
 P(\calE|Z=M,S=s,T=t)=P(\calEtil | \Ztil = \Mtil,S=s,T=t)
\end{align}

Note that this implies $P(\calE|Z=M)=P(\calEtil | \Ztil = \Mtil)$. Using this we can rewrite the quantitiy $R^{\beta}(M)$ in a similar way as $R^{\beta}(\Mtil)$:

\begin{align}
R^{\beta}(M) = \frac{\Tr(M\nu)}{\Tr(M\xi)},
\end{align}
where $\nu$ and $\xi$ are the \textit{same} matrices used to express $R^{\beta}(\Mtil)$ in equation (\ref{eq:mtilbias}).  In addition, we can lower-bound $\Tr(M\xi)$ and $\Tr(\Mtil\xi)$ as follows:
\begin{align}
\Tr(M\xi) & \geq  2\delta \cdot 2^{-m} \Tr(M) \geq   2\delta \cdot 2^{-m} \norm{M} \\
 &\geq   2\delta \cdot 2^{-m}, \\
\Tr(\Mtil\xi)& \geq  \delta \cdot 2^{-m} \Tr(\Mtil) \\ 
& \geq  \delta \cdot 2^{-m} \norm{\Mtil} \\
 &\geq  \delta \cdot 2^{-m} (1-\mu) \geq \delta \cdot 2^{-m} \cdot \tfrac{1}{2}.
\end{align}

Where we used that $M$ is $2\delta$-non-negligible, $\Mtil$ is $\delta-$non-negligible and the inequalities \eqref{eq:tr1} and \eqref{eq:tr2}.

Note that we use equation \eqref{eq:Mcloseness} and $\|M\|=1$ to get:

\begin{align}
 &\|M\| = \|M - \Mtil +\Mtil \| \leq \|M - \Mtil \| + \| \Mtil\| \\
 &\implies 1 \leq \mu + \| \Mtil \| \\
 &\implies \| \Mtil \| \geq 1 - \mu \label{eq:tr2}
\end{align}

We also used the fact that $\mu \leq \frac{2}{3} \cdot \delta \cdot 2^{-m}$, $\delta \leq \frac{1}{2}$ and $m\geq k \geq k_0 \geq 1$ (as defined in \cite{Liu15}) to lower bound $1-\mu$ as follows:
\begin{align}
 1-\mu \geq 1 - \frac{2}{3} \cdot \delta \cdot 2^{-m} \geq 1 - \frac{1}{3} \cdot 2^{-m} \geq 1 - \frac{1}{6} \geq \frac{1}{2}
\end{align}

Now we can write $R^{\beta}(M) - R^{\beta}(\Mtil)$ as follows:
\begin{align}
R^{\beta}(M) - R^{\beta}(\Mtil) &= \frac{\Tr(M\nu)}{\Tr(M\xi)} - \frac{\Tr(\Mtil \nu)}{\Tr(\Mtil \xi)} \\
&= \frac{\Tr(M\nu)}{\Tr(M\xi)} -\frac{\Tr(\Mtil \nu)}{\Tr(M\xi)}  + \frac{\Tr(\Mtil \nu)}{\Tr(M\xi)} - \frac{\Tr(\Mtil \nu)}{\Tr(\Mtil \xi)} \\
&= \frac{\Tr(M\nu)-\Tr(\Mtil \nu)}{\Tr(M\xi)}  + \frac{\Tr(\Mtil \nu)\Tr(\Mtil \xi)}{\Tr(M\xi)\Tr(\Mtil \xi)} - \frac{\Tr(\Mtil \nu) \Tr(M\xi)}{\Tr(M\xi)\Tr(\Mtil \xi)} \\
&= \frac{\Tr((M-\Mtil)\nu)}{\Tr(M\xi)} + \Tr(\Mtil\nu) \frac{\Tr((\Mtil-M)\xi)}{\Tr(M\xi)\Tr(\Mtil\xi)}.
\end{align}
We can then upper-bound this quantity:
\begin{align}
\begin{split}
\abs{R^{\beta}(M) - R^{\beta}(\Mtil)}
 &\leq \frac{\mu}{  2\delta \cdot 2^{-m}} + (1+\mu) \frac{\mu}{  2\delta \cdot 2^{-m} \cdot  \delta \cdot 2^{-m} \cdot \tfrac{1}{2}} \\
 &= \frac{\mu}{  2\delta \cdot 2^{-m}} \biggl( 1 + \frac{(1+\mu)}{ \delta \cdot 2^{-m} \cdot \tfrac{1}{2}} \biggr) \\
 &\leq 2\mu \biggl( \frac{2^m}{ \delta} \biggr)^2.
\end{split}
\end{align}
This completes the proof of Lemma~\ref{lem:continuity}.
\end{proof}

From Lemma~\ref{lem:continuity} we can show that Protocol~\ref{prot:OTIQM} is secure, when the adversary observes any separable measurement outcome $M \in W$ that is $2\delta$-non-negligible.

Note that $\norm{M}=1$ (that is assumed without loss of generality \cite{Liu15}) implies $\Tr(M) \geq 1$: 
\begin{align}
\Tr(M) &\geq 1 \label{eq:tr1}
\end{align}

Let $\Mtil \in \Wtil$ be the nearest point in the $\mu$-net $\Wtil$. Then we have:

\begin{align}
 \|M-\Mtil \|\leq \mu, \label{eq:Mcloseness}
\end{align}

where $\mu = 2^{- (\alpha/2)k}\frac{\delta^2}{4^m}$.

Then from Lemma~\ref{lem:3.6}, $\Mtil$ is $\delta$-non-neglibible. Then from equation \eqref{eq:provfg} we get that $|R^{\beta}(\Mtil)| \leq \lambda$, where $\lambda = 2^{-(\alpha/2)k}\cdot2r$:
 
\begin{align}
-2^{-\frac{\alpha}{2}k} \cdot 2r  \leq R^{\beta}(\Mtil) \leq 2^{-\frac{\alpha}{2}k} \cdot 2r.  \label{eq:add1}
\end{align} 

Using Lemma~\ref{lem:continuity} and substituting $\mu$ we get that:

\begin{align}
 \label{eq:37}
|R^{\beta}(M)-R^{\beta}(\Mtil)| \leq 2 \mu \left( \frac{2^m}{\delta} \right)^2 = 2 \cdot 2^{-(\alpha/2)k}  \\
\implies -2 \cdot 2^{-(\alpha/2)k} \leq R^{\beta}(M)-R^{\beta}(\Mtil) \leq 2 \cdot 2^{-(\alpha/2)k}.  \label{eq:add2}
\end{align}

By adding equations \eqref{eq:add1} and \eqref{eq:add2} we get:

\begin{align}
-2^{-(\alpha/2)k}\cdot 2(r+1) \leq R^{\beta}(M) \leq 2^{-(\alpha/2)k} \cdot 2(r+1)\\
\implies |R^{\beta}(M)| \leq 2^{-\frac{\alpha}{2}k} \cdot 2(r+1)
\end{align}

Then from Lemma~\ref{lem:Rxor} we see that Protocol~\ref{prot:OTIQM} is secure for all $2\delta$-non-negligible measurement outcomes $M \in W$ that a dishonest user Bob may observe:

\begin{align}
 \| \pr_{ \beta(A_0, A_1) \calE|Z=M} - \pr_U\| \leq 2^{-\frac{\alpha}{2}k} \cdot 2(r+1) + \eps = 2^{-\frac{\alpha}{2}k} \cdot 2(r+1) + 2^{-\eps_0 k}  \leq 2^{-\Omega(k)}. \label{eq:secbound}
\end{align}

Consider any LOCC adversary, and let $Z$ be the random variable representing the measurement outcome. We can then write:

\begin{align}
\begin{split}
 \|\pr_{\beta(A_0, A_1) Z}& - \pr_U \times \pr_Z \| \\
  & \leq \sum_{M} P(Z=M)\|\pr_{\beta(A_0, A_1) |Z=M} - \pr_U \|\\
    & \leq 2 \delta + \sum_{M: M\text{ is $2\delta$-non-negligible}} P(Z=M)\|\pr_{\beta(A_0, A_1) |Z=M} - \pr_U \|,\label{as1}
  \end{split}
\end{align}

since $\sum_{M: M \text{is $2 \delta$-negligible}} P(Z=M) \leq 2\delta$.

Taking into account the bound shown in equation~\eqref{eq:secbound} and the fact that $P(\neg \calE | Z=M) \leq \eps$, equation \eqref{as1} becomes

\begin{align}
\begin{split}
 \|\pr_{\beta(A_0, A_1) Z}& - \pr_U \times \pr_Z \| \\
  & \leq 2 \delta + 2 \eps + \sum_{M: M\text{ is $2\delta$-non-negligible}} P(Z=M)\left( \|\pr_{\beta(A_0, A_1), \, \calE |Z=M} - \pr_U \|  \right) \\
  & \leq 2 \delta + 2 \eps +  2^{-\frac{\alpha}{2}k} \cdot 2(r+1) + \eps \\ 
  & \leq 2\cdot 2^{-\delta_0 k} + 3\cdot 2^{-\eps_0k} +  2^{-\frac{\alpha}{2}k} \cdot 2(r+1). 
  \end{split}
\end{align}

  Note that in the last step we use the definitions of $\delta =  2^{-\delta_0 k} $ and $\eps= 2^{-\eps_0k}$.

Then Theorem~\ref{thm4.5} with $\ell' = \kappa k$, where $0 < \kappa < \min\left\{\frac{\delta_0}{2},\frac{\eps_0}{2},\frac{\alpha}{4}\right\}$ and

\begin{align}
 \|\pr_{\beta(A_0, A_1) Z} - \pr_U \times \pr_Z \| \leq 2\cdot 2^{-\delta_0 k} + 3\cdot 2^{-\eps_0k} +  2^{-\frac{\alpha}{2}k} \cdot 2(r+1) = \frac{\eps'}{2^{2\ell'+1}},
\end{align}

implies that

\begin{align}
\begin{split}
\label{eq70}
  \| \pr_{A_{1-D}\, A_{D} \, D \, Z}  - \pr_{U^\ell} \times \pr_{A_{D} \ \,D \, Z} \| & \leq \eps' \\
  &\leq  2^{2\ell'+1} \cdot \left( 2\cdot 2^{-\delta_0 k} + 3\cdot 2^{-\eps_0k} +  2^{-\frac{\alpha}{2}k} \cdot 2(r+1) \right) \\
    &\leq  2^{2\ell'+1} \cdot \left( 2\cdot 2^{-\delta_0 k} + 4\cdot 2^{-\eps_0k} +  2^{-\frac{\alpha}{2}k} \cdot 2(r+1) \right) \\
  & \leq 2^{-\delta_0 k+2\ell'+2} + 2^{-\eps_0k+2\ell'+3} +  2^{-\frac{\alpha}{2}k+2\ell'+2} + 2^{-\frac{\alpha}{2}k+2\ell'+2+\ln{r}} \\
  & \leq 2^{-(\delta_0 k-2(\ell'+1))} + 2^{-(\eps_0k-2\ell'+3)} +  2^{-(\frac{\alpha}{2}k-2(\ell'+1))} 	\\
  &\; \; \; \, + 2^{-(\frac{\alpha}{2}k-2(\ell'+2+\theta \ln{k})-\ln{(\gamma+1)})}
 \end{split}
\end{align}

Next we examine the term  $2^{-(\delta_0 k-2(\ell'+1))}$, since $\ell' < \frac{\delta_0}{2}k$ then:
\begin{align}
 \exists c>0 \; \exists k'_0 : \forall k>k'_0\; & \text{ the following holds: }  \delta_0 k-2(\ell'+1) \geq c k \\
 \implies& \delta_0 k-2(\ell'+1) \in \Omega{(k)}
\end{align}

then we have that for sufficiently large $k$:

\begin{align}
2^{-(\delta_0 k-2(\ell'+1))} \leq 2^{-\Omega{(k)}} \label{eq:lreqq1}
\end{align}

In a similar way we get that since $\ell'<\frac{\eps_0}{2}k$ then for sufficiently large $k$:

\begin{align}
2^{-(\eps_0 k-2\ell'+3)} \leq 2^{-\Omega{(k)}} \label{eq:lreqq2}
\end{align}

Finally since $\ell'< \frac{\alpha}{4}k$ then for sufficiently large $k$

\begin{align}
2^{-(\frac{\alpha}{2} k-2(\ell'+1))} \leq 2^{-\Omega{(k)}} \label{eq:lreqq3}
\end{align}

and

\begin{align}
2^{-(\frac{\alpha}{2}k-2(\ell'+2+\theta \ln{k})-\ln{(\gamma+1)})} = 2^{-(\frac{\alpha}{2}k-2\ell' -o(k)) } \leq 2^{-\Omega{(k)}}, \label{eq:lreqq4}
\end{align}

 which holds since 
 
\begin{align}
 f(k)&=2\theta\ln{k}+4+\ln{(\gamma+1)} \in o(k) .
\end{align}

Thus from equations \eqref{eq:lreqq1}, \eqref{eq:lreqq2}, \eqref{eq:lreqq3} and \eqref{eq:lreqq4},  for sufficiently large $k$

\begin{align}
\| \pr_{A_{1-D}\, A_{D} \, D \, Z}  - \pr_{U^\ell} \times \pr_{A_{D} \ \,D \, Z} \| \leq 2^{-\Omega{(k)}},
\end{align}

which completes the proof of Theorem~\ref{thm:main} \qed
\chapter{Flavours Of Oblivious Transfer}
\label{Chapter4}
\lhead{Chapter 4. \emph{Flavours Of Oblivious Transfer}}

In the previous chapter, we showed that a secure string \ROT protocol can be constructed in the isolated qubits model.
In this chapter we use the \ROT functionality to construct protocols that implement more complex oblivious transfer functionalities.

First we present a protocol that implements the \OT functionality using an instance of the \ROT functionality. 
As we have already discussed in \cref{Chapter1}, an \OT protocol is sufficient to implement any two-party computation securely, which makes it a fundamental problem in cryptography.

Secondly, we present a reduction from \kOT and \kROT to a series of $k$ \OT{s}, that was first introduced in \cite{BCR86}. 

Finally, we construct a protocol that implements the weaker \sROT functionality using only $\log{k}$ \ROT functionalities. 

These results are more general as these protocols are not restricted to the isolated qubits model as they rely on the existence and composability of a secure \ROT protocol in a cryptographic model.

\section{\OT from \ROT}

In this section, we introduce a protocol that implements the \OT functionality making use of a \ROT functionality. A sketch of Protocol~\ref{prot:new-OT} can be seen in \fref{img:new-OT}

\begin{prot}
 \label{prot:new-OT}
 A \OT protocol between user Alice with inputs $A_0,A_1 \in \{0,1\}^\ell$ and user Bob with input $D \in \{0,1\}$.
 
 \begin{enumerate}
  \item Alice  and Bob  use a \ROT functionality with no input and input $D$ respectively.
  \item Alice receives outputs $S_0,S_1 \in \{0,1\}^\ell$ and Bob receives $S_D \in \{0,1\}^\ell$.
  \item Alice then sends two messages $Y_0,Y_1$ such that:
  \begin{align}
   Y_0 &=S_0 \oplus A_0 \\
   Y_1 &=S_1 \oplus A_1
  \end{align}
 \item Bob then receives output:
  \begin{align}
   X_D = Y_D \oplus S_D
  \end{align}
 \end{enumerate}
\end{prot}

\begin{figure}[H]
 \centering
 \def\svgwidth{0.8\textwidth}
  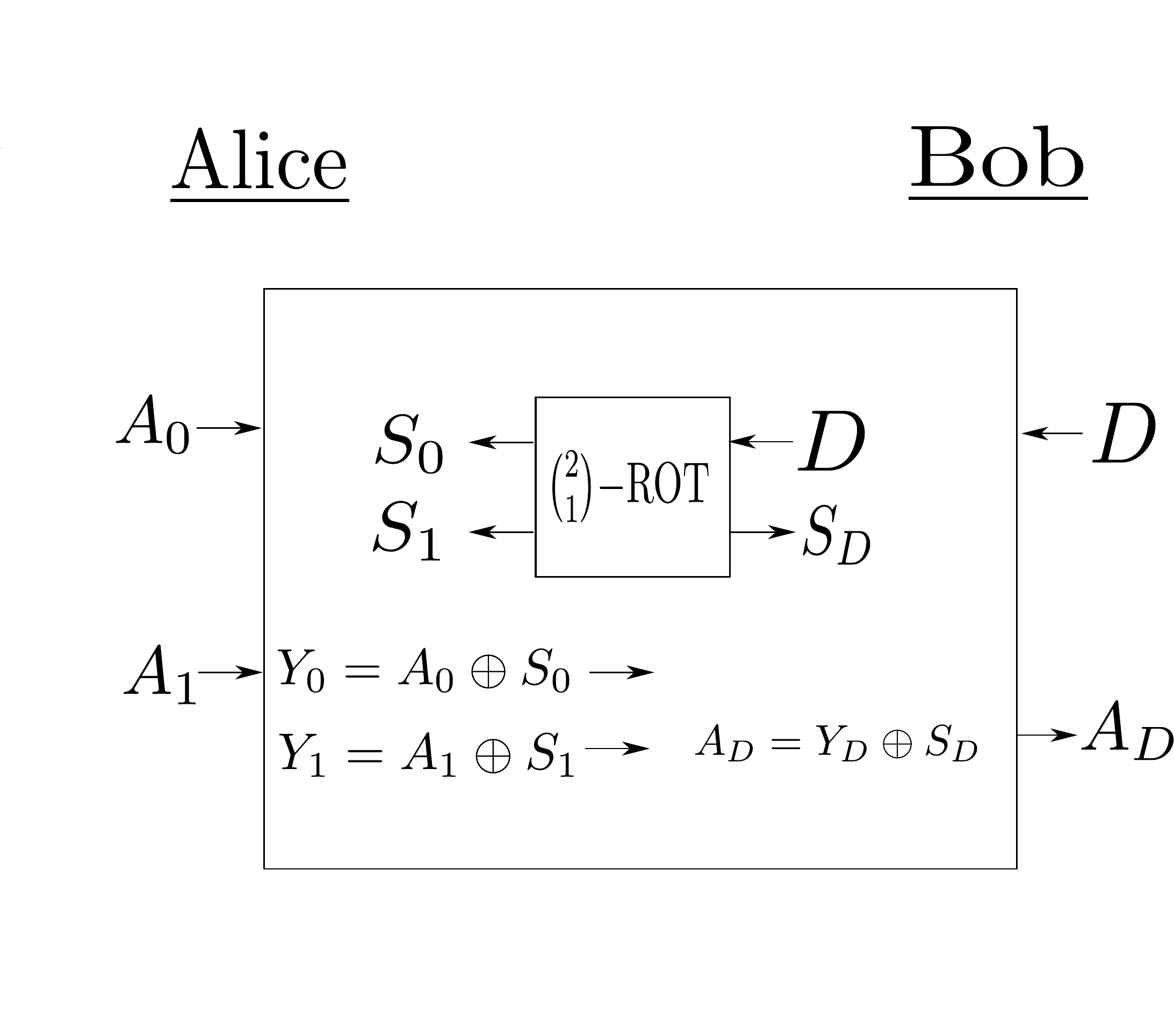
   \caption{Sketch of the \OT Protocol \ref{prot:new-OT} using a \ROT functionality.}
  \label{img:new-OT}
\end{figure}

We introduce the theorem that states that if the \ROT functionality is implemented securely, then so is the \OT functionality. 

\begin{theorem}
\label{thm:new-OT}
If the \ROT functionality used in Protocol~\ref{prot:new-OT} fulfills the $\eps-$security Definition~\ref{def:secrot}, then the \OT Protocol~\ref{prot:new-OT} is $\eps-$secure according to Definition~\ref{def:secot}.
\end{theorem}

\subsection{Proof of Theorem \ref{thm:new-OT}}

In order to prove Theorem~\ref{thm:new-OT} we need to show that all the conditions of Definition~\ref{def:secot} are fulfilled.

\subsubsection{Correctness}
 
 \begin{proof}
 
 If both Alice and Bob are honest, they follow Protocol~\ref{prot:new-OT}. 
 
 Then if the \ROT functionality is implemented correctly Alice will receive outputs $S_0,S_1$ and Bob will receive $S_D$ except with probability $\eps$.
 
 After Alice sends $Y_0,Y_1$, where $Y_i = A_i \oplus S_i$, Bob outputs $Y_D \oplus S_D = A_D$, implying correctness except with probability $\eps$.

 \end{proof}

 \subsubsection{Security for Alice}
 
 \begin{proof}
 For an honest user Alice, security of the \ROT functionality implies that there exists a random variable $D'$ such that
 
\begin{align}
 \pr_{S_{1-D'}G'S_{D'}D'} \appe \pr_{U} \cdot \pr_{G'S_{D'}D'} \label{eq:ld5},
\end{align}
 
 where $S_0,S_1$ are the outputs of honest Alice and $G'$ is the output of a dishonest user Bob.

 Define random variable $D''=D'$, then since an honest user Alice does not use her inputs $A_0,A_1$ in the \ROT it is clear that
 
 \begin{align}
  \pr_{G'D''S_0S_1A_0A_1} = \pr_{G'D''S_0 S_1} \cdot \pr_{A_0A_1} \label{eq:ld3},
 \end{align}

which implies that

\begin{align}
 \pr_{D''A_0 A_1} = \pr_{D''} \cdot \pr_{A_0A_1} \label{eq:ldA}.
\end{align}

Furthermore \eref{eq:ld3} implies that 

\begin{align}
   \pr_{G'D''S_{D''}A_0A_1} = \pr_{G'S_{D''}|D''} \cdot \pr_{D''} \cdot \pr_{A_0A_1},
\end{align}

and from \eref{eq:ldA} we have that

\begin{align}
  \pr_{G'D''S_{D''}A_0A_1} = \pr_{G'S_{D''}|D''} \cdot \pr_{A_0A_1D''},
\end{align}
  
  and 
  
\begin{align}
  \pr_{G'D''S_{D''}A_0A_1} = \pr_{G'S_{D''}|D''A_{D''}} \cdot \pr_{A_{D''}A_{1-D''}D''} \label{eq:ld4}.
\end{align}

 Bob's input $G''$ depends on $G',Y_0,Y_1$, where $Y_0 = S_0 \oplus A_0$ and $Y_1 = S_1 \oplus A_1$.
 
 Then from \eref{eq:ld5} and \eqref{eq:ld3}  we get
 
 \begin{align}
  \pr_{S_{1-D''}G'D''S_{D''}A_{D''}} \appe \pr_{U} \cdot \pr_{G'S_{D''}D''} \cdot \pr_{A_{D''}},
 \end{align}

 which implies that
 
  \begin{align}
  \pr_{S_{1-D''}|G'D''S_{D''}A_{D''}} \appe \pr_{U}.
 \end{align}
 
 Then $Y_{1-D''}$ is independent of $A_{1-D''}$ given $G',D'',S_{D''},A_{D''}$ and $Y_{D''}=A_{D''}\oplus S_{D''}$.
 
 Therefore

 \begin{align}
 \pr_{Y_0Y_1 G' D'' S_{D''}A_{D''}A_{1-D''}} \appe \pr_{Y_0Y_1| G' D'' S_{D''}A_{D''}} \cdot \pr_{ G' D'' S_{D''}A_{D''}} \cdot \pr_{A_{1-D''}} 
 \end{align}

 and taking into account \eref{eq:ld4}

 \begin{align}
 \pr_{Y_0Y_1 G' D'' S_{D''}A_{D''}A_{1-D''}} &\appe \pr_{Y_0Y_1| G' D'' S_{D''}A_{D''}} \cdot \pr_{ G' S_{D''}|D'' A_{D''}} \cdot \pr_{D'' A_{D''} A_{1-D''}} \\
 \Rightarrow \pr_{ G'' D'' A_{D''}A_{1-D''}} &\appe  \pr_{G'' |D''A_{D''}} \cdot \pr_{D''A_{D''}A_{1-D''}} \label{eq:ldB}
 \end{align}

 Thus, equations \eqref{eq:ldA} and \eqref{eq:ldB} imply that Protocol~\ref{prot:new-OT} is secure for Alice.
 \end{proof}

 \subsubsection{Security for Bob}
 \begin{proof}
 If the \ROT functionality is secure for honest user Bob with input $D$, there exist random variables $S'_0,S'_1$ such that
 
 \begin{align}
  \pr[G=S'_D] \geq 1 - \eps
 \end{align}

 and

 \begin{align}
  \pr_{S'_0 S'_1 D} \appe \pr_{S'_0 S'_1} \pr_{D} \label{eq:ld1},
 \end{align}

 where $D$ and $G$ are Bob's input and output used in $F_{\text{\ROT}}$.
 
 We can then define random variables $A'_0 = Y_0 \oplus S'_0$ and $A'_1 = Y_1 \oplus S'_1$, where $Y_0$ and $Y_1$ are the messages sent by Alice to Bob after the \ROT.

 Then from \eref{eq:ld1} and since Alice receives no further information from Bob after the \ROT has been used it is clear that
 
 \begin{align}
  \pr_{A'_0 A'_1 D} \appe \pr_{A'_0 A'_1} \pr_{D}.
 \end{align}

 Finally since Bob is honest, his output $G'$ will be $G'=G \oplus Y_D$, which implies that
 
  \begin{align}
  \pr[G=A'_D] \geq 1 - \eps.
 \end{align}

 Thus Protocol~\ref{prot:new-OT} is secure for Bob.

\end{proof}
This completes the proof of Theorem~\ref{thm:new-OT}.

\section{\kOT And \kROT From \OT}

The following \kOT protocol makes use of $k$ \OT functionalities and was first presented in \cite{BCR86}.

\begin{prot}
\label{prot:ineff-kOT}
 A \kOT protocol between user Alice with inputs $X_1,\dots,X_k \in \{0,1\}^\ell$ and user Bob with input $D \in \{1,\dots,k\}$.

\begin{enumerate}
 \item  Alice chooses strings $B_1,\dots, B_{k} \in \{0,1\}^\ell$ uniformly at random. 
 \item Alice inputs  $A_{1,0}= B_1$ and $A_{1,1}=X_1$ in the first \OT. Bob inputs his choice $D_1=\delta_{1,D}$ and receives string $A_{1,D_1}$.
\item  For $i=2,\dots,k$: \begin{enumerate}
                                 \item Alice inputs strings  $A_{i,0} = B_{i} \oplus B_{i-1}$ and $A_{i,1} = X_i \oplus B_{i-1}$ in the $i^{\text{th}}$ \OT.
                                 \item Bob inputs his $i^{th}$ choice $D_i = \delta_{i,D}$ and receives the string $A_{i,D_i}$
                           \end{enumerate}
\item Bob receives output:   \begin{align}   X_D = A_{D,1}\oplus \left( \bigoplus_{\substack{j=1}}^{D-1} A_{j,0} \right)	   \end{align}	
\end{enumerate}
\end{prot}

It is easy to see that a \kROT protocol can be constructed if Alice chooses her input messages $X_1,\dots, X_k$ uniformly at random. 
A sketch of the security proof for this protocol can be found in \cite{BCR86}. We present this protocol and its possible extension to a \kROT protocol to argue that indeed such a protocol can be constructed
from a secure \OT. However, it requires $k$ (or~$k-1$)\footnotemark \OT functionalities. In the next section we present a protocol that fulfills a weaker security definition but requires only $\log{k}$ \ROT{s}. We will
later use that protocol in \cref{Chapter5} to achieve secure password-based identification.

\footnotetext{This protocol can be implemented using $k-1$ \OT{s} if in the last \OT Alice inputs $A_{k,0}= B_{k-1} \oplus X_{k-1}$ and $A_{k,1} = B_{k-1} \oplus X_k$. }

\section{\sROT from \ROT}

\subsection{Protocol And Security Definition}

In this section, we introduce a protocol that implements the \sROT functionality. 
While the following \sROT protocol fulfills a weaker security definition, it
is more efficient than Protocol~\ref{prot:ineff-kOT} or its extension to a \kROT as it makes use of only $\log{k}$ \ROT{s} instead of $k$ \OT{s}.

Alice with no input, receives $\log{k}$ pairs of strings $(A_{i,0},A_{i,1})$ from the $i^{\text{th}}$ \ROT, for $i~\in~\{1,\dots,\log{k}\}$. 
Her output messages $S_1,\dots,S_k$ will later be composed of the possible additions of these strings, for example $S_1 = \bigoplus_{i=1}^{\log{k}}A_{i,0}$.

Bob with input $D \in \{1,\dots,k\}$, that can be seen as a string $\{D_1|D_2|\dots|D_{\log{k}}\}$, in turn inputs the $i^{\text{th}}$ bit of his choice $D_i$ to the $i^{\text{th}}$ \ROT functionality and obtains
output $A_{i,D_i}$.
He finally adds the outputs he received to obtain his output of the \sROT funcitonality $S_D~=~\bigoplus_{i=1}^{\log{k}}A_{i,D_i}$. A sketch of the protocol can be seen in \fref{img:sROT}.

\begin{prot}
\label{prot:sROT}
Sender-randomised \sROT protocol between user Alice with no input and user Bob with input $D=\{D_1|D_2|\dots|D_{\log{k}}\} \in \{1,\dots,k\}$.

\begin{enumerate}
\item  For $i=1,\dots,\log{k}$: \begin{enumerate}
                                 \item Alice with no input receives strings $A_{i,0}, A_{i,1} \in \{0,1\}^\ell$ as outputs of the $i^{\text{th}}$ \ROT,
                                 \item Bob inputs his $i^{\text{th}}$ choice $D_i$ and receives the string $A_{i,D_i}$.
                                \end{enumerate}
\item Alice receives outputs: 
\begin{align}
 \begin{array}{cl}
   S_1 &= A_{1,0} \oplus A_{2,0} \oplus \dots \oplus A_{\log{k},0} \\
   S_2 &= A_{1,1} \oplus A_{2,0} \oplus \dots \oplus A_{\log{k},0} \\
   \vdots& \hspace{0.6cm} \vdots \hspace{3cm}  \vdots\\
   S_k &= A_{1,1} \oplus A_{2,1} \oplus \dots \oplus A_{\log{k},1} 
 \end{array}
\end{align}

and Bob receives output:
\begin{align}
S_D = \bigoplus_{i=1}^{\log{k}}A_{i,D_i} 
\end{align}
	
\end{enumerate}
\end{prot}

\begin{figure}[H]
 \centering
 \def\svgwidth{0.8\textwidth}
  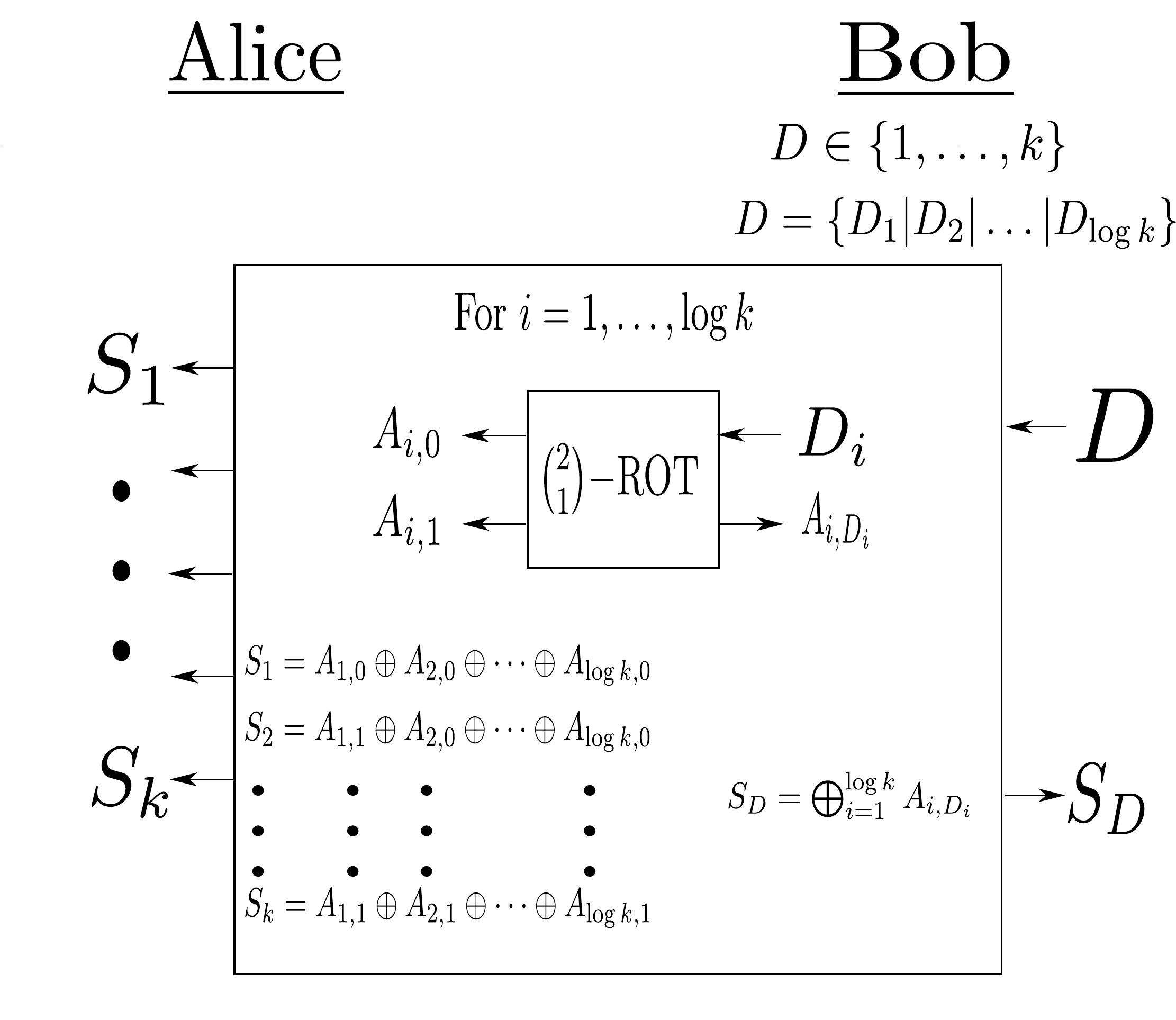
   \caption{Sketch of the \sROT Protocol \ref{prot:sROT} using $\log{k}$ \OT functionalities.}
  \label{img:sROT}
\end{figure}

We now introduce the theorem that states that if the \ROT functionalities used in the above protocol are secure, in the sense of Definition~\ref{def:secot}, then the protocol implements the \sROT functionality 
securely, according to the Definition~\ref{def:seckrot}.

\begin{theorem}
\label{thm:sROT}
If the \ROT functionalities used in Protocol~\ref{prot:sROT} are $\eps-$secure, according to Definition~\ref{def:secrot}, then Protocol~\ref{prot:sROT} is $\eps'-$secure according to Definition~\ref{def:secsrot}, where $\eps' \leq \eps \log{k}$.
\end{theorem}

\subsection{Proof Of Theorem~\ref{prot:sROT}}
The following proof consists of three parts as we have to show that all three requirements of Definition~\ref{def:seckrot} hold.

\subsubsection{Correctness}

First we show that the correctness requirement of Definition~\ref{def:secsrot} holds if the correctness requirement of Definition~\ref{def:secrot} holds.

\begin{proof} 
 We show that for two honest users who follow Protocol~\ref{prot:sROT}, the protocol is correct if the \ROT functionality is implemented correctly.
 
 An honest user Alice receives outputs $(A_{i,0}, A_{i,1})$ from the $i^{\text{th}}$ \ROT used in the protocol, for $i=1,\dots,\log{k}$. 
 She receives outputs $S_1\dots S_k$ by modulo 2 addition of these inputs, for example $S_1 = \bigoplus_{i=1}^{\log{k}}A_{i,0}$.
 
 Then Bob with input $D=\{D_1|\dots|D_{\log{k}}\}$, uses $D_i$ as input to the $i^\text{th}$ \ROT. 
 If the \ROT{s} are correct he receives $A_{i,D_i}$ except with probability $\eps$ for all $i=1,\dots,\log{k}$.

 He then correctly computes $X_D=\bigoplus_{i=1}^{\log{k}}A_{i,D_i}$ except with probability $\eps' \leq \eps \cdot \log{k}$. 
 
 Thus Protocol~\ref{prot:sROT} is correct.
\end{proof}

\subsubsection{Security For Alice}

Secondly we show that if the \ROT are secure for Alice in the sense of Definition~\ref{def:secrot} then the security for Alice condition of Definition~\ref{def:secsrot} holds.

\begin{proof}
 For an honest user Alice with no input and any dishonest user Bob with output $G'$
 we can define a random variable $D'$ such that $D'=\{D'_1|\dots|D'_{\log{k}}\}$, where $D'_i$ is Bob's input to the $i^{\text{th}}$ \ROT functionality used in the protocol.

Then for all $I=\{I_1|\dots|I_{\log{k}} \}$ such that $I \neq D'$, there exists at least one $j \in \{1,\dots,\log{k}\}$ such that $I_j \neq D'_j $.
 
Since the \ROT{s} are secure for Alice for any dishonest user Bob with output $G'_j$ 
   there exists a random variable $D'_i$ such that 
 
 \begin{align}
  \pr_{A_{j,I_j} G'_j A_{j,D'_j} D'_j} \appe \pr_{U}  \cdot \pr_{G'_j A_{j,D'_j} D'_j} \label{eq:123}.
 \end{align}
 
 Then since $S_{D'} = \bigoplus_{i=1}^{\log{k}} A_{i,D'_i}$ and  
 $S_I = \bigoplus_{j=1}^{\log{k}} A_{j,I_j}$ for any $I \neq D'$,

 \begin{align}
  \pr_{S_{I}G'S_{D'}D'} \appee \pr_U \cdot \pr_{G'S_{D'} D'} \label{eq:124},
 \end{align}
 
 where $\eps' \leq \eps \cdot \log{k}$.

 Equation \eqref{eq:124} proves that Protocol~\ref{prot:sROT} is secure for Alice according to Definition~\ref{def:secsrot}.
 
\end{proof}

\subsubsection{Security For Bob}

Finally we show that if the \ROT functionalities are secure for Bob according to Definition~\ref{def:secrot} then Protocol~\ref{prot:sROT} is secure for Bob, in the sense of Definition~\ref{def:secsrot}.

\begin{proof}
For an honest user Bob with input $D=\{d_1|\dots|d_{\log{k}}\} \in \{1,\dots,k\}$ and any dishonest user Alice we define random variables $S'_1, \dots S'_k$ in the following way:
 
 \begin{align}
 S'_I = \bigoplus_{j=1}^{\log{k}}A_{j,I_j} \text{ , for } I \in \{1,\dots,k\},
 \end{align}

 where $A'_{i,0}$ and $A'_{i,1}$ are the inputs in the $i^{\text{th}}$ \ROT used in Protocol~\ref{prot:sROT}.
 
 Then since the \ROT{s} are secure for Bob there exist random variables $A'_{i,0}, A'_{i,1}$ such that for the output of the $i^{\text{th}}$ \ROT 

 \begin{align}
  \pr[G_{i_{\text{\OT}}}=A'_{i,D_i}] \geq 1-\eps, \label{sobsrot1}
 \end{align}
 
 and 
 
 \begin{align}
  \pr_{D_i A'_{i,0} A'_{i,1}} \appe \pr_{D_i} \cdot \pr_{A'_{i,0}A'_{i,1}}, \label{sobsrot2}
 \end{align}
 
  for all $i=1,\dots,\log{k}$.
 
 Since $S_D = \bigoplus_{j=1}^{\log{k}} A'_{j,D_j}$, \eqref{sobsrot1} implies that
 
 \begin{align}
  \pr[G_{\text{\sROT}}=S'_D] \geq 1-\eps',
 \end{align}

 with $\eps' \leq \eps \cdot \log{k}$.
 
 Furthermore \eqref{sobsrot2} implies that the distribution of $D_i$ is independent of the inputs of that \OT. 
 Since a dishonest user Alice receives no information the choices of Bob are independent. Then  since $D=\{D_1|\dots|D_{\log{k}}\}$ and $S'_I = \bigoplus_{j=1}^{\log{k}}A'_{j,I_j}$ 
 
  \begin{align}
  \pr_{DS'_1\dots S'_k} \appee \pr_{D} \cdot \pr_{S'_1\dots S'_k}
 \end{align}

\end{proof}

Thus if the \OT functionalities used are correct, then Protocol~\ref{prot:sROT} is secure, this concludes the proof of Theorem~\ref{thm:sROT}.

\chapter{Secure Identification}
\label{Chapter5}
\lhead{Chapter 5. \emph{Secure Identification}}

In this chapter, we aim to construct a protocol that achieves secure password-based identification. 
In order to do so, we first study existing protocols, namely the protocol proposed in \cite{DFSS07}, that achieves secure identification in the bounded quantum storage model. 

First, we adapt this protocol to the isolated qubits model by using a \kOT functionality. As we have seen in \cref{Chapter4}, it is possible to construct a \kOT protocol in this model.
However, we notice that this identification protocol requires interaction from Bob to Alice.

Secondly, we study if it is possible to construct a non-interactive secure identification protocol. We show that such a protocol is impossible to construct, even based on oblivious transfer. 

Finally, we prove the security of an interactive password-based identification protocol that makes use of a \sROT functionality. 
The latter can be implemented efficiently, as we showed in \cref{Chapter4}.

 \section{Secure Identification From \kOT}
 
  \label{sec:IID}
  There are a number of secure password-based identification protocols in the literature, we present \cite[Protocol~Q-ID]{DFSS07} that achieves secure identification in the bounded quantum storage model.
  Let $c:\mathcal{W} \mapsto \{+,\times\}^n$ be the encoding function, where $+$ is the computational and $\times$ is the Hadamard basis. 
  
  \begin{prot}
  \label{prot:SchIID} 
  Interactive Password-based Identification with inputs $W_\A$ and $W_\B$, the passwords of user Alice and user Bob respectively. Let $\mathcal{F}$ and $\mathcal{H}$ be families of strong 2-universal hash functions \cite{DFSS07}:
  \begin{enumerate}
   \item The user Alice picks $x_R \in \{0,1\}^n$ and $\theta \in_R \{+,\times\}^n$ she then sends state $\ket{x}_\theta$ to Bob
    \item Bob measures $\ket{x}_\theta$ in basis $D=c(W_\B)$. Let $X_D$ be the outcome.
   \item Alice picks $f \in \mathcal{F}$ uniformly at random and sends $\theta$ and $f$ to Bob. Both compute $I_W :=\{i: \theta_i = c(W)_i\}$.
   \item Bob picks $h \in \mathcal{H}$ uniformly at random and sends $h$ to Alice.
   \item Alice sends $z:=f(X_{W_\A}|_{I_{W_\A}}) \xor h(W_\A)$ to Bob, where $X_{W_\A}|_{I_{W_\A}}$ is the restriction of $X_{W_\A}$ to the coordinates $X_i$ with $i \in I_{W_\A}$.
   \item Bob accepts if and only if $z=f(X_D|_{I_{W_\B}}) \xor h(W_\B)$
  \end{enumerate} 
  \end{prot}
  
  While this protocol is secure in the bounded-quantum storage model, it is not secure in the isolated qubits model as we have discussed in \cref{Chapter2}.
  Note however that the first part of the protocol (steps 1-4) can be seen as a protocol that implements the \kOT functionality. 
  As we have shown in \cref{Chapter4}, there exists a protocol that achieves that in the isolated qubits model.
   Taking this fact into account, we construct a password-based identification protocol that relies on the security of a \kOT functionality. A sketch of the protocol is presented in \fref{img:secidkOT}.

\label{proof}

\begin{prot}
\label{prot:secid-kOT}
Password-based identification protocol with inputs $W_\A$ and $W_\B$, the passwords of user Alice and user Bob respectively.
Let $\mathcal{H}$ be a family of strong 2-universal hash functions such that $h\in\mathcal{H}$ and $h:\{1,\dots,k\} \mapsto \{0,1\}^\ell$. 
Then the protocol between Alice and Bob is the following:
\begin{enumerate}

\item The user Alice uses a \kOT functionality, $\F_{\text{\kOT}}$, with inputs $X_1,X_2,\dots,X_k \in \mathcal{X}$.

\item The user Bob inputs his choice $D=W_{\B}$ to the \kOT and receives the message $X_{D}$

\item Bob sends a function $h \in \mathcal{H}$ to Alice.

\item Alice sends $z:=X_{W_\A} \xor h(W_\A)$ to Bob.

\item The user Bob outputs 1 if $z=X_D \xor h(W_\B)$ and 0 otherwise.

\end{enumerate}
\end{prot}

In \sref{sec:secid-srot}, we will prove the security of a similar protocol that relies on the weaker \sROT functionality, that can be implemented efficiently using $\log{k}$ \ROT functionalities.

\begin{figure}
 \centering
  \def\svgwidth{0.8\textwidth}
  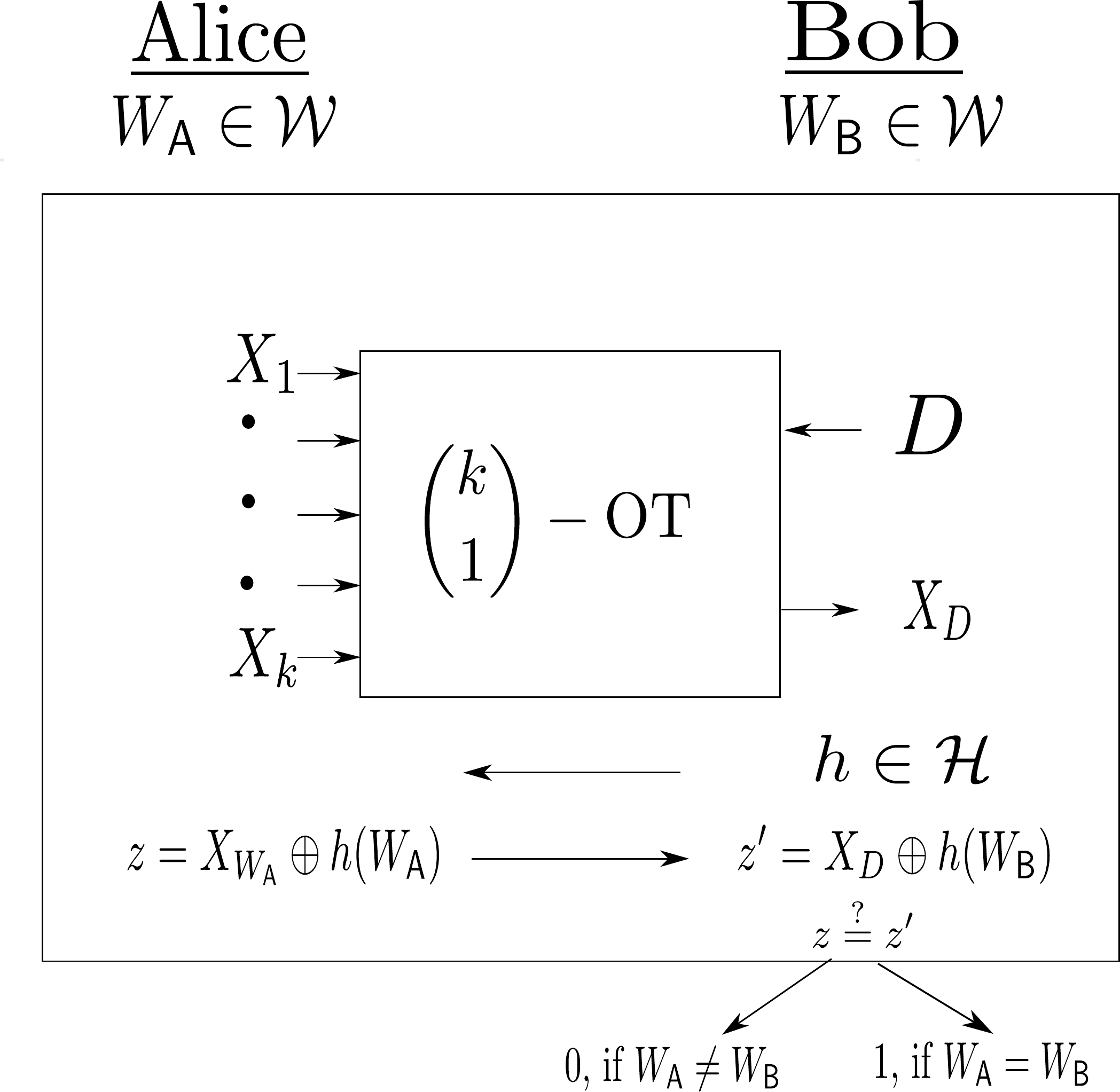
 \caption{Sketch of the password-based identification protocol (Protocol~\ref{prot:secid-kOT}) that makes use of a \kOT functionality.}
   \label{img:secidkOT}
\end{figure}

  Note however, that both Protocol~\ref{prot:SchIID} and \ref{prot:secid-kOT} use interaction from Bob to Alice.
  But if we assume that we can implement a \kOT functionality securely and non-interactively,
  can we use it to construct a non-interactive secure identification protocol?
  
  \section{Impossibility Proof}
  \label{sec:impossibility}

  In this section we study if it is possible to construct a non-interactive password-based identification protocol using a \kOT.
   We introduce a general protocol that uses one instance of the $F_{\text{\kOT}}$ functionality to implement the identification functionality $F_{ID}$ 
   and we then prove that such a protocol cannot be secure.
  We aim to emphasize the importance of the interaction from Bob to Alice (step 4 of Protocol~\ref{prot:secid-kOT}) in order to implement the identification functionality securely.

  \subsection{Non-Interactive Password-Based Identification}
  
  We formally introduce the protocol later (Protocol~\ref{prot:NIID}), but we first describe it to give some intuition and argue why the protocol is a general form for all such possible protocols. 
  The user Alice has as input a password $W_{\A}$, can choose inputs $X_1, \dots , X_k$ to the \kOT and sends some extra information $Y$ to the user Bob depending on the specific protocol. 
  In Protocol~\ref{prot:secid-kOT} for example, $Y$ is the function $f$ and the message $z=f(X_{W_\A} \oplus h(W_\A)$.
  The user Bob has as input his password $W_{\B}$ and makes a choice regarding the message he will retrieve from the \kOT. His choice $D$ may depend on his password so that
  if his password choice is equal to the choice of Alice he will be able to correctly identify her while he will not be able to do so in any other case. So his choice
  is described by a deterministic function $c:\mathcal{W}\times\mathcal{Y} \mapsto \{1,\dots,k\}$ such that $D=c(W_{\B},Y)$ returns the message that when combined with the information $Y$ will
  allow him to check if $W_{\A} = W_{\B}$. 
  
  We note that any non-interactive protocol that uses the $\F_{\text{\kOT}}$ functionality once has to be of this form. Since it is non-interactive the user Alice can use the functionality of $\F_{\text{\kOT}}$ once and at most send some additional information $Y$. On the other hand the user Bob receives information $Y$ and can interact with the $\F_{\text{\kOT}}$ functionality by inputing his choice $D$, that can at most depend on both $Y$ and his password choice $W_\B$. Finally, he can, at most, use $W_\B$, $D$, $X_D$ and $Y$ as inputs to some function $g$ to evaluate the equality function.
 
 \begin{prot}
 \label{prot:NIID}
Non-interactive identification protocol with inputs $W_\A$ and $W_\B$, the passwords of user Alice and user Bob respectively :
\begin{enumerate}

\item The user Alice uses the non-interactive \kOT functionality, $\F_{\text{\kOT}}$, with inputs \break
$X_1,X_2,\dots,X_k~\in~\mathcal{X}$ and sends additional information $Y$ to Bob. She chooses the inputs and additional information
uniformly at random from a joint distribution $\pr_{X_1 \dots X_k Y|W_\A}$.

\item The user Bob inputs his choice $D=c(W_{\B},Y)$ to the \kOT 
 and receives the message $X_{D}$.

\item The user Bob then computes and outputs the acceptance predicate, G: 

\begin{align}
 G=g(W_{\B},X_D,Y) = \left\{ 
 \begin{array}{cc}
 0 \text{, if he rejects,} \\ 1 \text{, if he accepts.}
 \end{array} \right.
\end{align}

\end{enumerate}

\end{prot}

A non-interactive identification protocol in this \kOT  hybrid model is defined by the following ingredients:
$\pr_{Y X_1 \dots X_k | W_\A} , \pr_{D|W_\B Y} , \pr_{G|W_\B DX_D Y}$

\label{error-case}
In order for Protocol~\ref{prot:NIID} to be secure it must fulfill the conditions of the security definition Definition~\ref{def:secid}.
We consider the special case for $\epsilon=0$ for perfect security of the protocol.
We did not study the case that a non-interactive $\epsilon-$secure password-based identification protocol can be constructed using oblivious transfer.
Although studying the $\epsilon>0$ case remains an interesting problem for future research, 
we consider the intuition we collect from the following proof (\sref{proof:niid}) sufficient to emphasise the importance of interaction between Bob and Alice as discussed in \sref{interaction}.
This result justifies the use of interaction in the construction of a secure password-based identification protocol, which is the main goal of this thesis.

Then for users Alice and Bob that hold $X_1,\dots,X_k,Y$ and $D,X_D,Y,G$ respectively, we can formulate the following security definition, that is equivalent to Definition~\ref{def:secid}.

\begin{definition}
\label{def:secNIID}
The non-interactive identification Protocol \ref{prot:NIID} is secure if the following conditions are fulfilled:
\begin{description}
  
\item [Correctness:] For honest user Alice and honest user Bob, Bob outputs $G=1$ if $W_\A = W_\B$.

\item [Security for Alice:] For any dishonest user Bob, for any distribution of $W_{\A}$, there exists a random variable
$W'$ that is independent of $W_{\A}$ and such that:

\begin{align}
\pr_{W_{\A} W' Y X_{D}|W' \neq W_\A } = \pr_{W_{\A}\lr W' \lr Y X_D|W' \neq W_{\A} }
\end{align}

\item [Security for Bob:] For any dishonest user Alice, for any distribution of $W_{\B}$, there exists a random variable $W'$
independent of $W_{\B}$ such that if  \break $W' \neq W_\B$ then $\pr[G=1|W_\B\neq W']=0,$ and:

\begin{align}
\pr_{W_{\B} W' Y X_1\dots X_k|W'\neq W_\B} = \pr_{W_{\B} \lr W'\lr Y X_1\dots X_k|W' \neq W_{\B} }
\end{align}

\end{description}
\end{definition}

The following theorem that states that it is impossible for a protocol that uses one instance of a \kOT functionality to implement the identification functionality securely.
  
\begin{theorem} 
\label{thm:NIID}
  If Protocol \ref{prot:NIID} is correct and secure for Alice according to Definition~\ref{def:secNIID}, then it is not secure for Bob. 
\end{theorem}

\subsection{Proof Of Theorem~\ref{thm:NIID}}  
\label{proof:niid}

We first introduce some lemmas that we will use later to prove Theorem~\ref{thm:NIID}.

\begin{lemma} 
\label{lem:ind}
 If Protocol \ref{prot:NIID} is secure for Alice then for all $i \in \{1,\dots,k\}$ the joint distribution of the random variables $X_i$ and $Y$ are independent of $W_\A$ .
\end{lemma}

\begin{proof}
 
 Since Protocol \ref{prot:NIID} is secure for Alice, for all $\pr_{W_\A}$, for all $i \in \{1,\dots,k\}$ there exists $W'$ independent of $W_\A$ such that:
 
 \begin{align}
 \pr_{W_\A W' X_i  Y  ,\, W' \neq W_\A} &=  \pr_{W_\A \lr W' \lr X_i Y | W' \neq W_\A}.
 \end{align}

 Then by definition:
 
 \begin{align}
 \pr_{W_\A|W' X_i  Y  ,\, W' \neq W_\A} &= \pr_{W_\A|W' , \, W' \neq W_\A}   \label{eq:l1pr1}
 \end{align}

 We also note that trivially when $W'=W_\A$,   
 
 \begin{align}
 \pr_{W_\A|W' X_i Y  ,\,W_\A = W'} = \pr_{W_\A | W' ,\, W_\A=W'} \label{eq:l1pr2} 
 \end{align}
 
 Using the property of the marginal distribution: 
  
 \begin{align}
  \pr_{W_\A|W'X_iY} =&  \pr [ W'\neq W_\A ] \pr_{W_\A|W'X_i Y  ,\,W'\neq W_\A}  \nonumber \\
     & +  \pr [ W' = W_\A ] \pr_{W_\A|W'X_i Y ,\, W' = W_\A} \\
   \stackrel{\text{\eqref{eq:l1pr1},\eqref{eq:l1pr2}}}{=}& \pr [ W'\neq W_\A ] \pr_{W_\A|W'  ,\,W'\neq W_\A} +  \pr [ W' = W_\A ] \pr_{W_\A|W' ,\, W' = W_\A} \\
   =&\pr_{W_\A|W'} \label{eq:l1pr3a}
 \end{align}
 
 Using the fact that $W'$ is independent of $W_\A$ equation  \eqref{eq:l1pr3a} becomes:
 
 \begin{align}
    \pr_{W_\A|W'X_iY} = \pr_{W_\A}  \label{eq:l1pr3}
 \end{align}

  From equation \eqref{eq:l1pr3} we observe that $W', X_i, Y$ are independent of $W_\A$.
  \end{proof}

 The next lemma states that if Protocol \ref{prot:NIID} is secure for Alice and correct then the function $c(\cdot,y)$ is injective for all possible $y$.

\label{Proof: c(W,Y) is injective}
\begin{lemma} \label{lem:inj}
If Protocol \ref{prot:NIID} is correct and secure for Alice then for the function $c: \mathcal{W}\times \mathcal{Y} \rightarrow [k]$ the following holds:
  \begin{align} 
 \forall y\in \mathcal{Y} \;\text{ with } \pr_Y(y)>0,\; \; \; \; c(\cdot,y)  \text{ is injective.}
  \end{align}

\end{lemma}
    
\begin{proof}

  Let us assume that the function  $c(W_\B,Y)$ is not injective. 
  
  Then $\exists y : \pr_Y(y) >0$ and $\exists j,m \in \mathcal{W}$ with $j \neq m$ such that 
  
  \begin{align}
c(j,y)=c(m,y). \label{eq:inj1}   
  \end{align}

Then clearly,
  
    \begin{align}
      X_{c(j,y)}=X_{c(m,y)},
     \end{align}
     
 which immediately implies that,
 
  \begin{align}
  g(j,X_{c(m,y)},y) &\stackrel{\eqref{eq:inj}}{=}g(j,X_{c(j,y)},y) \label{eq:inj}. 
  \end{align}

 Let us assume that $W_\A=j$ and $W_\B=m$. Since Protocol \ref{prot:NIID} is correct, Bob computes:
 
  \begin{align}
 g(m,X_{c(m,y)},Y) =0,
  \end{align}
  
  but also
  
  \begin{align}
  g(j,X_{c(m,y)},y) &\stackrel{\eqref{eq:inj}}{=}g(j,X_{c(j,y)},y)= 1 \label{eq:insecurity}.
  \end{align}

  This means that for $W_\A \neq W_\B$, Bob learns the password of Alice and thus 
  
\begin{align}
\pr_{W_{\A} W' Y X_{D}|W' \neq W_\A } \neq \pr_{W_{\A}\lr W' \lr Y X_D|W' \neq W_{\A} },
\end{align}

  which means that Protocol \ref{prot:NIID} is not secure for Alice.
 Thus if Protocol \ref{prot:NIID} is correct and secure for Alice the function $c(W_\B,Y)$ is injective. 
 \end{proof}

  From correctness we expect that for all password inputs $w \in \mathcal{W}$ thre exists a $y \in \mathcal{Y}$ and there exists a $x \in \mathcal{X}$ that Alice can input in the \kOT and will lead Bob to output $G=g(w,x,y)=1$.
  The following lemma states that it must be so for all $y \in \mathcal{Y}$ with $\pr_Y(y)>0$, for all password inputs $w \in \mathcal{W}$ simultaneously.
  Intuitively this is so because otherwise a dishonest user Bob would gain some information on the password of Alice from the message $Y$. He could for example exclude some password choices after seeing $Y$,
   making the protocol insecure for Alice. 
  
  \begin{lemma} \label{lem:possiblex}
   If Protocol \ref{prot:NIID} is correct and secure for Alice, then for all $w \in \mathcal{W}$, for all $y \in \mathcal{Y}$ such that $\pr_Y(y)>0$  there exists a $x \in \mathcal{X}$ such that: 
   \begin{align}
    g(w,x,y)=1
   \end{align}
  \end{lemma}

\begin{proof}
  We will prove this lemma by contraposition.
  
 Assume that  there exists a $w \in \mathcal{W}$ and there exists a $y \in \mathcal{Y}$ with $\pr_{Y|W_\A}(y|w)>0$ such that for all $x \in \mathcal{X}$:
 \begin{align}
      g(w,x,y)=0
 \end{align}

 Let $W_\A=W_\B=w$. Then for all $x \in \mathcal{X}$ 
 
 \begin{align}
  g(w,x,y)=0,
 \end{align}

 which implies that Protocol \ref{prot:NIID} is not correct.
 
 So far we have shown that if Protocol \ref{prot:NIID} is correct, then for all $w\in \mathcal{W}$, for all $y \in \mathcal{Y}: \pr_{Y|W_\A}(y|w)>0$ there exists a $x \in \mathcal{X}$ such that $g(w,x,y)=1$.
 
 Furthermore security for Alice implies that $\pr_{Y|W_\A} = \pr_Y$ via Lemma \ref{lem:ind} and thus we can conclude that:
 
 For all $w \in \mathcal{W}$ for all $y \in \mathcal{Y}$ with $\pr_Y(y)>0$, there exists a $x \in \mathcal{X}$ such that:

 \begin{align}
  g(w,x,y)=1
 \end{align}

\end{proof}

 Note that on the one hand, the information $Y$ Alice sends does not give any information about her password input, which is necessary to ensure her security. On the other hand,
 it also means that Alice does not commit to a password choice by sending the information $Y$ to Bob.

We now prove Theorem~\ref{thm:NIID} using the above lemmas. The intuition behind the following proof is that if the identification Protocol \ref{prot:NIID} is correct and secure for Alice, a dishonest Bob
cannot learn anything about the password of Alice from the output of the \kOT $X_D$ or the additional information $Y$ alone except for the output $G$.
He also does not learn anything about the other inputs in the \kOT, thus allowing a dishonest user Alice to launch an attack by choosing
the inputs to the \kOT, $X_1, \dots , X_k$, such that each one of them combined with $Y$ will force Bob to accept for all of his password choices $W_\B$. 
Then Protocol~\ref{prot:NIID} is clearly not secure for Bob since the dishonest user Alice does	not need to choose one password $W_\A$ but can force Bob to always accept.

\begin{proof}
 If Protocol \ref{prot:NIID} is correct and secure for Alice, a dishonest user Alice can use the following attack to force Bob to accept for all of his password choices.
 Alice inputs $W_\A=1$, chooses a value for $Y$ honestly and then picks  the inputs to the \kOT, $X_1,\dots,X_k$, such that for every password choice $W_\B$ of the user Bob,
 he will obtain $X_i=X_{c(W_\B,Y)}$ such that he will output $G=1$.
 
\paragraph{Attack Strategy Of Dishonest User Alice}
 \begin{enumerate}
 \item Alice chooses $Y=y$ (honestly) according to the distribution $\pr_{Y|W_\A=1}$, and sends it to Bob.
 
 \item Alice uses the non-interactive \kOT functionality, with inputs $X_1,\dots,X_k$ that she chooses as follows.
 
 For every password $w \in \mathcal{W}$:
 \begin{enumerate}
  \item Find a $x$ such that: 
    \begin{align}\pr_{X_j|W_{\A}=w,Y=y}(x)>0 \text{, with } j=c(w,y)\end{align}
   and 
 \begin{align}G= g(w,x,y)=1.\end{align} 
  \item Set input $X_j = x \text{, with } j=c(w,y)$. 
 \end{enumerate}
\end{enumerate}

 Note that step 2 is possible because correctness and security for Alice imply, via Lemma \ref{lem:possiblex}, that for all possible choices of $Y$ and for all possible password choices $w \in \mathcal{W}$ there exists  a $x \in \mathcal{X}$ such that $G=1$.

Furhtermore, Lemma~\ref{lem:inj} implies that the function $c(W_\B,y)$ is injective for all $y \in \mathcal{Y}$. 
Then once $y$ is chosen, for every $w \in \mathcal{W}$ there exists only one $j \in \{1,\dots,k\}$, such that $j=c(w,y)$.

 These two facts allow a dishonest Alice to choose the inputs of the \kOT, such that for every password choice of Bob $w \in \mathcal{W}$ he retrieves a $x \in \mathcal{X}$ such that he outputs $G=1$.

 In more detail, after receiving the \kOT and $Y=y$, the honest user Bob chooses a password $W_\B$, inputs his choice $D=c(W_\B,y)$ to the \kOT and receives the message $X_D$.
 As described above for every one of his password choices he receives a message $X_D=x$ such that $g(w,x,y)=1$.
 He then outputs $G=g(W_\B,X_D,Y)=1$ for any of his password choices, which implies that Protocol \ref{prot:NIID} is not secure for Bob.

  \end{proof}
  
 \subsection{The Importance Of Interaction}
  \label{interaction}
  
  Theorem~\ref{thm:NIID} shows that a non-interactive protocol using one instance of a \kOT functionality cannot implement the identification functionality securely. We proved that 
  security for Alice and correctness of the protocol allow the attack described above to succeed and we claim that this is true as long as Alice has knowledge of the function $g(W_\B,X_D,Y)$.
  This knowledge allows her to choose the inputs to the \kOT such that Bob will accept for all of his passwords, making the protocol insecure.
  
  It is then interesting to examine where (interactive) protocols that are known to be secure differ. In the examples of \cite[Protocol~Q-ID]{DFSS07} and Protocol~\ref{prot:secid-kOT} the user Bob sends some information to Alice after receiving the \kOT,
  in this case a strong 2-universal hash function. This interaction makes the protocol secure for Bob against the above attack, because the (possibly extended) function $g(W_\B,X_D,Y)$ as defined above, is 
  fixed after Alice has chosen her inputs to the \kOT. 
  Fixing the function $g$ after Alice has commited to her inputs to the \kOT denies her the possibility to choose them in such a way that $G=1$ for all passwords.
  Thus we conclude that as long as the function that is used by Bob to determine the acceptance predicate is fixed before Alice chooses her inputs to the \kOT, she can choose the inputs such that the above attack will work.
  
  Extending the non-interactive Protocol \ref{prot:NIID} by allowing multiple uses of the \kOT functionality, or even \kbOT functionalities from Alice to Bob will still be insecure for Bob as the function $g(\cdot)$ is fixed before Alice
  chooses her inputs to the \kOT{s}. 
  For the same reason, an interaction from Bob to Alice before she chooses her inputs to the \kOT would not stop the above attack from functioning.

\section{Secure Identification From \sROT With Interaction}
\label{sec:secid-srot}

In \sref{proof} we introduced Protocol~\ref{prot:secid-kOT}, a password-based identification protocol based on a \kOT functionality. 
The \kOT construction we showed in \cref{Chapter4} is, however, inefficient as it requires $k$ instances of an \OT functionality, where $k$ is the number of passwords. 

In this section, we show that one can instead use the \sROT functionality that only requires $\log{k}$ 
\OT{s}. While weaker than the \kOT or \kROT functionalities, it is sufficient to achieve security for the password-based identification protocol.

We introduce the password-based identification protocol that relies on the security of a \sROT functionality. A sketch of the protocol is shown in Figure~\ref{img:secidsROT}

  \begin{prot}
 \label{prot:secid-sROT}
Password-based identification protocol with inputs $W_\A, W_\B \in \mathcal{W}$, the passwords of user Alice and user Bob respectively, where $\mathcal{W} = \{1,\dots,k\}$. 
Let $\mathcal{H}$ be a family of strong 2-universal hash functions such that $h\in\mathcal{H}$ and $h: \mathcal{W} \mapsto \{0,1\}^\ell$. 
Then the protocol between Alice and Bob is the following:

\begin{enumerate}
\item The users Alice and Bob employ a \sROT functionality, $\F_{\text{\sROT}}$ that takes no input from Alice and input $D \in \mathcal{W}$ from Bob.

\item The user Alice receives outputs  $S_1,S_2,\dots,S_k \in \mathcal{S}$ and Bob receives output string $S_{D} \in \mathcal{S}$, where $\mathcal{S} =\{0,1\}^\ell$.

\item Bob chooses a function $h \in \mathcal{H}$ uniformly at random and sends $h$ to Alice.

\item Alice sends $z:=S_{W_\A} \xor h(W_\A)$ to Bob.

\item The user Bob accepts if and only if $z=S_D \xor h(W_\B)$ .

\end{enumerate}
\end{prot}

We now introduce a theorem that states that the secure identification protocol we propose above is secure in the sense of Definition~\ref{def:secid} 
if the \sROT functionality used is secure according to Definition~\ref{def:seckrot}.
 
 \begin{theorem}
  \label{thm:idsROT}
  If there exists a protocol that implements the \sROT functionality $\eps-$securely according to Definition~\ref{def:secsrot} and the min-entropy of password choices $W$ is $H_{min}(W) \geq 1$,
  then Protocol~\ref{prot:secid-sROT} is $\eps'-$secure in the sense of Definition~\ref{def:secid}, where $\eps' = \eps + \frac{k^2}{2^\ell}$.
 \end{theorem}

 In \cref{Chapter4} we have showed how to construct a secure \sROT protocol relying on a secure \ROT.
 Moreover, in \cref{Chapter3} we have showed that a secure \ROT protocol exists in the isolated qubits model.
 Taking these results into account, the previous theorem states that our password-based identification protocol is secure in the isolated qubits model.
 
 \begin{figure}[H]
 \centering
  \def\svgwidth{0.8\textwidth}
  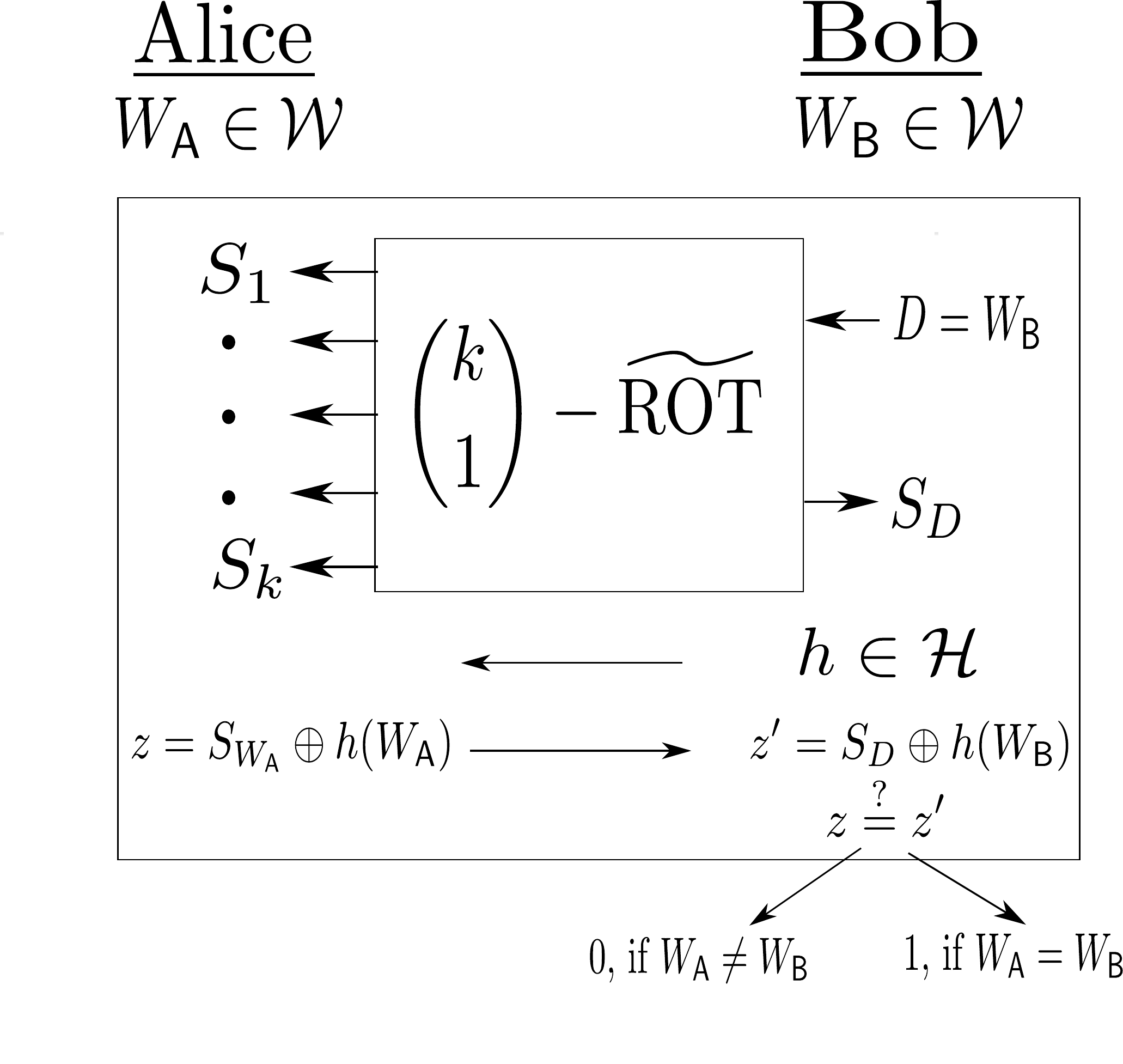
 \caption{Sketch of the password-based identification protocol (Protocol~\ref{prot:secid-sROT}) that makes use of a \sROT functionality.}
   \label{img:secidsROT}
\end{figure}

 \subsection{Proof Of Theorem~\ref{thm:idsROT}} 
 
 Finally, in this section, we prove that if the \sROT used in Protocol~\ref{prot:secid-sROT} is secure then the identification protocol is secure.
 
  \subsubsection{Correctness}
  
\begin{proof}
 Honest users Alice and Bob hold inputs $W_\A$ and $W_\B$. 
 
 Following Protocol~\ref{prot:secid-sROT} Bob inputs $D=W_\B$ in the \sROT functionality which is correct and thus Bob receives output
 
 \begin{align}
  G_{\text{\sROT}} = S_D = S_{W_\B},
 \end{align}

 except with probability $\eps$.

Then $z'=S_D\oplus h(W_\B)=S_{W_\B} \oplus h(W_\B)$.

Since Alice is honest $z=S_{W_\A} \oplus h(S_{W_\A})$ and thus if $W_\A = W_\B$  then $z=z'$ and Bob will output $G=1$. 

Similarly if $W_\A \neq W_\B$ Bob will output $G=0$ except with probability $\frac{k^2}{2^\ell}$, the probability that $h(W_\A)=h(W_\B)$ given $W_\A \neq W_\B$.

Thus Protocol~\ref{prot:secid-sROT} is correct except with probability $\eps'=\eps+\frac{k^2}{2^\ell}$.
 
\end{proof}

  \subsubsection{Security For Alice}
    \label{sec:prsecalice}

  \begin{proof}
    
    For honest user Alice with input $W_\A$ and any dishonest user Bob we can define a random variable $W'=D'$, where $D'$ is Bob's input
    in the \sROT used in the protocol. 
    Since Bob has received no information before he decides on his input to the \sROT, $D'$ and thus $W'$ are independent of the input $W_\A$ that honest Alice holds.
    
    Furthermore $W_\A$ is independent of the messages $S_1, \dots S_k$ that Alice receives from the \sROT since she has chosen it before receiving any output from the \sROT and the latter takes no input from Alice.
    
    From the above it is clear that if $W' \neq W_\A$ and since $W' = D'$ then
    
    \begin{align}
     \pr_{W_\A W' S_{W'} | W' \neq W_\A }= \pr_{W_\A \lr W' \lr S_{W'}|W'\neq W_\A} \label{lreq1}.
    \end{align}

   Furthermore from the security of \sROT, there exists a random variable $D'$ such that for all $I \neq D'$ the following holds
   
   \begin{align}
    \pr_{S_I D' S_{D'}|D' \neq I} \appe \pr_U \cdot \pr_{D' S_{D'} | D' \neq I} \label{lreq2}.
   \end{align}

   Then for $I=W_\A$ and $W_\A \neq W'$ equation \eqref{lreq2} becomes,
   
   \begin{align}
        \pr_{S_{W_\A} W' S_{W'}|W' \neq W_\A} \appe \pr_U \cdot \pr_{W' S_{W'}  | W' \neq W_\A} \label{lreq22}.
   \end{align}

   Consider the random variable $Z=S_{W_\A} \oplus h(W_\A)$ that describes the message Alice sends to Bob after receiving the hash function $h$.
   Taking into account that $W_\A$ is independent of $W',D'$ and $S_1,\dots,S_k$, including $S_{W'}$ and $S_{W_\A}$, conditioned on the event $W'\neq W_\A$
   Then from the properties of the modulo 2 addition and equation \eqref{lreq22} we have that 
   
   \begin{align}
    \pr_{ZW_\A W'S_{W'}|W'\neq W_\A} \appe \pr_{U} \cdot \pr_{W_\A W'S_{W'}|W'\neq W_\A}.
   \end{align}
    
    From \eref{lreq1} the above can be written as:
    
    \begin{align}
    \pr_{ZW_\A W'S_{W'}|W'\neq W_\A} &\appe \pr_{U} \cdot \pr_{W_\A \lr W' \lr S_{W'}|W'\neq W_\A} \\
    &\appe \pr_{W_\A \lr W' \lr Z S_{W'}|W'\neq W_\A} \\
    &\appee \pr_{W_\A \lr W' \lr Z S_{W'}|W'\neq W_\A},
   \end{align}
   
   with $\eps' = \eps + \frac{k^2}{2^\ell}$.
   
   Then Protocol~\ref{prot:secid-sROT} is $\eps'-$secure for Alice.
  \end{proof}

  \subsubsection{Security For Bob}
  \label{sec:prsecbob}

  \begin{proof}

  Since the \sROT functionality, that is implemented by an honest user Bob with input $D=W_\B$ and a dihonest user
  Alice, is secure for Bob, there exist random variables $S'_1, \dots , S'_k$ such that
  
  \begin{align}
   \pr_{DV'S'_1 \dots S'_k} \appe \pr_{D} \cdot \pr_{V'S'_1\dots S'_k}.
  \end{align}

It is clear that $V',S'_1,\dots,S'_k$ are independent of $W_\B$.

We then define  $Z_i = S'_i \oplus h(i)$ for $i \in \mathcal{W}$. Consider the event $\mathcal{E}$ that all $Z_i$'s are distinct.
Since $h$ is strong 2-universal and is also independent of $S'_i$ the $Z_i$'s are pairwise independent. Then from the union bound it follows that 
the  event $\mathcal{E}$ occurs except with probability $k(k-1)/2 \cdot 1/2^\ell \leq k^2/2^{\ell+1}$.

We define random variable $W'$ such that the message sent by Alice $Z= S_{W'}\oplus h(W')$. If $Z \neq Z_i$ for all $i$ then we set $W' = \bot$ and honest Bob always outputs $G=0$
regardless of his password choice $W_\B$. In this case a dishonest user Alice learns nothing about $W_\B$. Similarly from the way that $W'$ is defined Bob will output $G=1$ if $W' = W_\B$.

Note that from security of the \sROT functionality, and since $h$ is picked uniformly at random, $W'$ is independent of $W_\B$.
This further implies that $Z_1,\dots,Z_k,Z$ are also independent of $W_\B$. Moreover since the event $\mathcal{E}$ is determined by the $Z_i$'s it also holds that
$Z_1,\dots,Z_k,Z$ are independent from $W_\B$ conditioned on the event $\mathcal{E}$ and even
given $W'$ conditioned on $\mathcal{E}$ and $W' \neq W_\B$.

Now consider $Z_1,\dots,Z_k,Z,G$, if $W' \neq W_\B$ and event $\mathcal{E}$ then Bob outputs $G=0$ with probability $P[G=0|W' \neq W_\B, \mathcal{E}] = 1$.
Then  $Z_1,\dots,Z_k,Z,G$ are independent of $W_\B$ given $W'$ conditioned on the event $W'\neq W_\B$ and $\mathcal{E}$, that is:

\begin{align}
 \pr_{W_\B W' Z_1 \dots Z_k Z G |W' \neq W_\B, \mathcal{E}} \appe \pr_{W_\B \lr W' \lr Z_1 \dots Z_k Z G |W' \neq W_\B ,\mathcal{E}}.\label{lreq3}
\end{align}

We then define $p = \pr[\mathcal{E}|W'\neq W_\B]$ and $\overline{p} = \pr[\overline{\mathcal{E}}|W'\neq W_\B]$. 
Note that $\pr[\overline{\mathcal{E}}] \leq k^2/2^{\ell + 1}$.

Furthermore since $H_{min}(W) \geq 1$ it is easy to see that  $\pr[W'=W_\B]\leq\frac{1}{2}$.
Then

\begin{align}
 \overline{p} = \pr[\overline{\mathcal{E}}|W'\neq W_\B] = \frac{\pr[\overline{\mathcal{E}}]}{1-\pr[W'=W_\B]} 
  \leq 2 \pr[\overline{\mathcal{E}}] \leq \frac{k^2}{2^\ell}.
\end{align}

Note that $\overline{p}$ upperbounds the probability $\pr[G=1|W'\neq W_\B] \leq \overline{p} \leq \frac{k^2}{2^\ell} \leq \eps'$, where $\eps'= \eps + \frac{k^2}{2^\ell}$,
fulfilling the first condition for security.

From basic probability theory and using \eref{lreq3} :

\begin{align}
 \pr_{W_\B W' Z'_1 \dots Z'_k Z G |W' \neq W_\B} &= p \cdot \pr_{W_\B W' Z'_1 \dots Z'_k Z G |W' \neq W_\B, \mathcal{E}} + \overline{p} \cdot\pr_{W_\B W' Z'_1 \dots Z'_k Z G |W' \neq W_\B, \overline{ \mathcal{E}}} \\
& \appe p \cdot \pr_{W_\B \lr W' \lr Z'_1 \dots Z'_k Z G |W' \neq W_\B, \mathcal{E}} + \overline{p} \cdot\pr_{W_\B W' Z'_1 \dots Z'_k Z G |W' \neq W_\B, \overline{ \mathcal{E}}}
 \end{align}

 Finally note that $\mathcal{E}$ is independent of $W_\B$ and $W'$ and thus also when conditioned on $W' \neq W_\B$, then from conditional independence
 
 \begin{align}
  \pr_{W_\B \lr W' \lr Z'_1 \dots Z'_k Z G |W' \neq W_\B} = p \cdot \pr_{W_\B \lr W' \lr Z'_1 \dots Z'_k Z G |W' \neq W_\B, \mathcal{E}} + \overline{p} \cdot \pr_{W_\B \lr W' \lr Z'_1 \dots Z'_k Z G |W' \neq W_\B, \overline{ \mathcal{E}}}.
 \end{align}

 The distance between two probability distributions is upper-bounded by one by definition. Then so is the distance between $ \pr_{W_\B W' Z'_1 \dots Z'_k Z G |W' \neq W_\B, \overline{ \mathcal{E}}}$ and 
 $\pr_{W_\B \lr W' \lr Z'_1 \dots Z'_k Z G |W' \neq W_\B, \overline{ \mathcal{E}}}$. 
 
 Thus the  distance between $ \pr_{W_\B  W' Z'_1 \dots Z'_k Z G |W' \neq W_\B} $ and $ \pr_{W_\B \lr W' \lr Z'_1 \dots Z'_k Z G |W' \neq W_\B} $ is upper-bounded by $\overline{p} + \eps \leq k^2/2^{\ell} + \eps=\eps'$.
 
 Then since Alice's output $V'$ is defined by $Z_1, \dots, Z_k, Z, G$:
 
 \begin{align}
  \pr_{W_\B  W' V' |W' \neq W_\B} \appee \pr_{W_\B \lr W' \lr V' |W' \neq W_\B}.
 \end{align}

  Thus Protocol~\ref{prot:secid-sROT} is secure for Bob.
  \end{proof}

Thus Protocol~\ref{prot:secid-sROT} is secure. We note that this protocol is more efficient than the previous identification protocols, because of the more efficient construction of \sROT compared to the \kROT protocol presented in \cref{Chapter4}.
   
\chapter{Conclusions \& Discussion}
\label{Chapter6}
\lhead{Chapter 6. \emph{Conclusions \& Discussion}}

In this final chapter, we summarise our main results and conclusions. Then we open the discussion of these results in relation to current knowledge. 
Finally, we pose some of the questions that arise from this discussion and propose possible future steps.

\section{Conclusions \& Discussion}

In \cref{Chapter3} we presented a secure string \ROT in the isolated qubits model, using the ``leaky'' OTM presented in \cite{Liu14b}.
We note that our proof follows a similar path to the one used in \cite{Liu15}, but makes use of the notion of non-degenerate linear functions coupled with the results of 
\cite{DFSS06}. The resulting proof is simpler than the original and allows us to construct a secure string \ROT in the IQM. 
This comes at a cost, by using Theorem~\ref{thm4.5} security is achieved with an error $2^{2 \ell' +1}$ times larger than the one presented in \cite{Liu15}.
Fortunately, as shown in equation \eqref{eq70}, this factor does not influence the security result.
However, Theorem~\ref{thm:main} implies that the security parameter $k$ has to be of the order of $\ell'$ in order for the protocol to achieve security.

In \cref{Chapter4} we made the first attempt to study more complex two-party functionalities in this model.
We propose secure \OT, \kOT and \sROT protocols that rely on the security and composability of an \ROT functionality.
In order to guarantee composability of the \ROT protocol, we restricted the users to measure at the end of each protocol. 
These protocols can then be implemented in the isolated qubits model using an \ROT protocol that is secure in that model, such as Protocol~\ref{prot:OTIQM} presented in \cref{Chapter3}.

The question that then arises is if the aforementioned assumption is realistic. Since composability of protocols has not been studied in the isolated qubits 
this questions remains an open problem for further study and we will briefly discuss this in the next section.

Following that, in \cref{Chapter5} we address an interesting problem of secure two-party computation, the evaluation of the equality function. 
We present a protocol for secure password-based identification that uses a \sROT functionality, motivated by the protocols proposed in \cite{DFSS07}.
However, the results of \cref{Chapter4} and \cref{Chapter5} are not limited by the specific cryptographic model, they can be implemented in any model in which there exists a protocol that implements the \ROT functionality securely and in a composable way.

An interesting question that we encountered on the way is if it is possible to construct non-interactive secure password-based identification protocols. 
In Section~\ref{sec:impossibility}, we proved constructing such a protocol based on oblivious transfer is impossible.
The interaction from Bob to Alice must define the way he computes his output in order for the protocol to be secure against the attack we presented in Section~\ref{sec:impossibility}.

Moreover, we claim that this result is not restricted to the secure evaluation of the equality function but also applies to more (or even all) secure-function-evaluation problems.
For example in the similar problem of Yao's millionaire problem where Bob computes a different function, a dishonest user Alice still has the ability to predetermine the output for all of her inputs,
as long as Bob's function is not determined after she has commited to her inputs.

\section{Future Work}

This thesis studies the construction of a secure string \OT protocol if the users are restricted to operations on single qubits and classical communication
between them and gives examples of possible applications to construct more complex secure two-party computation protocols such as password-based identification. 
This leads of course to new questions that remain an open challenge for the future. 

As we mentioned in the last section, studying if composability holds in the IQM is likely the most interesting problem that arises from this thesis.
If it is shown to be so, we have shown that a secure \OT construction is possible in the IQM, which would imply that any secure two-party computation functionality can be implemented.
If however composability does not hold in the IQM, then constructing and analysing protocols in this model would prove an exciting challenge in itself. 
For example, the problem of analysing the security of two parallel \OT{s} and modelling the measurement strategies of an adversary who is allowed to partially measure qubits from the first and second \OT and adapt his
measurement strategy depending on partial results of each \OT seems to be a first challenge for further research.

As we have already described in \cref{Chapter1} there are numerous results that prove the impossibility of oblivious transfer in a fully quantum world. 
Nevertheless, there exist different approaches to restrict the users in a realistic fashion and achieve oblivious transfer. 
One of the most interesting questions that arises from this train of thought is to find the minimal and most realistic restrictions or assumptions needed to achieve secure \OT{s}. 
For example, Liu has the question of allowing a number of entangling operations on the isolated qubits in \cite{Liu15},  which could be a possible approach to generalise the isolated qubits model.

We have discussed in more detail two approaches to limit an adversary, restricting him to single-qubit operations or restricting his qubit-storage capacity. 
So far, existing protocols that are secure against one type of adversary are not secure against the other. The question then is, could we construct protocols that combine the power of these two models? For example using
a \OT that is secure in the IQM and one that is secure in the noisy-storage model to construct one \OT that is secure in both models and using the modulo 2 addition of their outputs?
This would then mean that the adversary would need to both have larger qubit storage capacities, in order to break the noisy-storage model \OT security and be able to perform entangling operations
on the qubits he receives, in order to break the isolated qubits model \OT security.

As a further approach to combine these models, Liu has addressed the question of allowing a number of entangling operations on the isolated qubits in \cite{Liu15}.
This could be the first step to define a more general model and should be investigate further.

One further possibility for future endeavours that arises from our impossibility proof in \sref{proof:niid} is to examine if our result indeed applies to more non-interactive two-party protocols.
However, there exist results that state that quantum one-time programs can be constructed from one-time memories \cite{BGS13}. 
We conjecture that there is a lower-bound on the number of one-time memories needed to construct a secure one-time program for password-based identification so that both results hold. 
Unfortunately we did not to study this into more detail in this thesis and we leave it as an open question.

Furthermore, we leave the task of extending the impossibility proof to the error case as discussed in \sref{error-case} as an open problem for the future.
Our intuition is that the attack described in the proof should function since interaction from Bob to Alice seems to be necessary to achieve security for Bob. Nevertheless formalising this intuition is an interesting extension of the 
impossibility proof discussed in this thesis. 


\addtocontents{toc}{\vspace{2em}} 

\appendix 



\chapter{Probability Theory} 

\label{AppendixA} 

\lhead{Appendix A. \emph{Probability Theory}} 

\section{Probability Theory}

\subsection{Random Variables}



The probability distribution of a random variable $X$ that takes values $x \in \Xcal$ is a function $\pr_X: \Xcal \mapsto [0,1]$ and is defined as:

\begin{align}
 \pr_X(x):=\pr[X=x] , \forall x \in \Xcal
\end{align}

Note that for every probability distribution the following holds:
\begin{align}
\sum_{x \in \Xcal} \pr_{X}(x)=1
\end{align}

The joint probablity distribution of two random variables $X$ and $Y$ that take values $x \in \Xcal$ and $y \in \Ycal$ respectively, is defined as:

\begin{align}
 \pr_{XY}(x,y) := \pr[X=x,Y=y],
\end{align}

and indicates the probability that $X$ takes the value $x$ and $Y$ takes the value $y$ simultaneously.

Let $\pr_{XY}$ be the joint distribution of random variables $X$ and $Y$. Then the distribution of $X$ can be obtained by marginalising over $Y$. The distribution $\pr_X$ is then called a marginal distribution:

\begin{align}
 \pr_{X}(x)=\sum_{y\in \Ycal} \pr_{XY}(x,y) \; \; \forall x \in \Xcal.
\end{align}

Let $\pr_{XY}$ be the joint of random variables $X$ and $Y$. If $X$ and $Y$ are \emph{independent} random variables the joint distribution can be written as:

\begin{align}
 \pr_{XY}(x,y) = \pr_X(x) \cdot \pr_Y(y) \; \; \forall x \in \Xcal,\;\; \forall y \in \Ycal.
\end{align}


Furthermore, the \emph{conditional probability distribution} of a random variable $X$ takes the value $x \in \Xcal$ given the event that the random variable $Y$ takes the value $y \in \Ycal$
is defined as:

\begin{align}
 \pr_{X|Y}(x|y):= \frac{\pr_{XY}(x,y)}{\pr_Y(y)}.
\end{align}

%
%
%

Moreover we introduce the symbol $\pr_{X \lr Y \lr Z}$, as used in \cite{DFSS07} and \cite{FS09}, to denote that the distribution of a random variable $X$ is 
independent of a random variable $Z$ given a random variable $Y$:

\begin{align}
 \pr_{X|YZ} = \pr_{X|Y}
\end{align}

Then we write:

\begin{align}
 \pr_{XYZ} = \pr_{X \lr Y \lr Z}
\end{align}

This notation is extended to $\pr_{XYZ|\cal{E}} = \pr_{X \lr Y \lr Z | \cal{E}}$
to denote that the distribution of a random variable $X$ is independent of a random variable $Z$ given a random variable $Y$ conditioned on an event $\cal{E}$:

\begin{align}
 \pr_{X|YZ\cal{E}} = \pr_{X|Y\cal{E}}
\end{align}

\paragraph{Boole's inequality}

The union bound or Boole's inequality states that the probability of at least one event occuring cannot be greater than the sum of the probabilities of all individual events.

Formally for a set of events $A_1, A_2, \dots $ the following inequality holds:

\begin{align}
 \pr \left( \bigcup_i A_i \right) \leq \sum_i \pr(A_i)
\end{align}


Finally, the expected value of a random variable $X$ that takes values $x \in \Xcal$ is defined as:

\begin{align}
\EE(x) = \sum_{x\in \Xcal} x\cdot \pr_X(x) 
\end{align}

\subsection{Uniform Distribution}

If a random variable $X$ is uniformly distributed it means that all of its values are equiprobable. 

\begin{definition}
 A random variable $X$ that takes values $x \in \Xcal$ is uniformly distributed if its distribution $\pr_X$ is of the following form:
 \begin{align}
  \pr_X(x) = \frac{1}{|\Xcal|} \; \; \; \forall x \in \Xcal
 \end{align}
 Then $\pr_X$ is a uniform distribution over $\Xcal$.
\end{definition}

%
%
%

\subsection{$\epsilon$-Net}

Intuitively an $\epsilon$-net is a subset of some normed space such that for every point of the original space there is some point in the 
$\epsilon$-net that is $\epsilon$-close to it. We now introduce the formal definition of an $\epsilon$-net. 

\begin{definition}
 Let $E$ be a subset of some normed space, with norm $\|\cdot\|$ and let $\epsilon >0$. Then $\Etil$ is an $\epsilon$-net for $E$ if $\Etil \subset E$, and for all $x \in E$, there exists some $y \in \Etil$
 such that:
 \begin{align}
  \|x-y\| \leq \epsilon
 \end{align}
\end{definition}


\chapter{Measures of Uncertainty} 

\label{AppendixB} 

\lhead{Appendix B. \emph{Measures of Uncertainty}} 

\section{Renyi Entropy}

\begin{definition} For a random variable $X$ that takes values $x \in \Xcal$, for $\alpha \in \RR$ with $ \alpha \geq 0 $ and $\alpha \neq 1$, the Renyi entropy of order $\alpha$ is defined as

 \begin{align}
  H_\alpha(X):= \frac{1}{1-\alpha} \log\left(\sum_{x \in \Xcal} \pr_X(x)^\alpha \right)
 \end{align}

\end{definition}

We note that the Renyi entropy is a generalised entropy.

For $\alpha = 1$ we obtain the Shannon entropy:

 \begin{align}
  H_1(X):= - \sum_{x \in \Xcal}\pr_X(x) \log \left( \pr_X(x) \right)
 \end{align}

For $\alpha = 0$ we obtain the max-entropy:

 \begin{align}
  H_0(X):= \log|\Xcal|
 \end{align}
 
\section{Min-Entropy}

One important measure of uncertainty for information theory is the Renyi entropy we get for
 $\alpha \rightarrow \infty $, namely the min-entropy:

 \begin{align}
  H_\infty(X):= \min_{x\in\Xcal} \left[- \log \pr_X(x) \right]
 \end{align}

It is the smallest of the Renyi entropies of order $\alpha$ and thus the most conservative estimate of uncertainty in a random variable. This is the reason
why it is widely used in cryptography. 

Similarly one can define the conditional min-entropy

 \begin{align}
  H_\infty(X|Y):=  \min_{x\in\Xcal,\; y \in \Ycal} \left[ - \log \pr_{X|Y}(x|y) \right]
 \end{align}

\section{Smoothed Min-Entropy}

The smoothed min-entropy defined below can be understood as the entropy of a distribution $\pr_X$ that is smoothed by cutting a probability mass $\eps$ from the largest probabilities.

\begin{align}
  H_\infty^\eps(X):=  \max_{\calE:\pr(\calE) \geq 1-\eps} \min_{x\in\Xcal} \left[ - \log \pr_{X}(x) \right]
 \end{align} 

Informally one can think of it as the maximum min-entropy available in any distribution that is $\eps$-close to the distribution $\pr_X$.
 
Furthermore the smoothed conditional min-entropy is defined as:

\begin{align}
  H_\infty^\eps(X|Y):=  \max_{\calE:\pr(\calE) \geq 1-\eps} \min_{x\in\Xcal,\; y \in \Ycal} \left[ - \log \pr_{X|Y}(x|y) \right]
 \end{align} 

The latter is an important measure in cryptography as it defines the maximum amount of randomness that is available from $X$ 
given $Y$ and $S$, except with probability $\eps$.


\chapter{Linear Algebra} 

\label{AppendixC} 

\lhead{Appendix C. \emph{Linear Algebra}} 

\section{Norms}

For any matrix $A \in \CC^{m \times n}$ and vector $x \in \CC^n$ we define the following norms:

Operator norm

\begin{align}
 \|A\| = \max_{\|x\|=1}\|Ax\|
\end{align}

%
%
%

Trace norm

\begin{align}
 \|A\|_{\Tr} = \Tr(\sqrt{AA^\dagger})
\end{align}
 
%
%

\paragraph{Statistical Distance}

Let $P$ and $Q$ be two probability distributions of a random variable $X$ that takes values $x \in \Xcal$.
Then the $\ell_1$ distance between them is defined as:

\begin{align}
 \|P-Q\| = \sum_{x} \left|\pr(x) - Q(x)\right|.
\end{align}

This is commonly called the statistical distance and is used as a distance measure between the two probability distributions.

%
%
%
%
%
%
%
%
%

\addtocontents{toc}{\vspace{2em}} 

\backmatter


\label{Bibliography}

\lhead{\emph{Bibliography}} 
\bibliographystyle{alphaarxiv}
\bibliography{quantum}

\end{document}